%% file: article.tex
\documentclass[12pt]{article}

\input{preamble.ltx}

\newif\ifanonymous
\anonymousfalse  

\ifanonymous
\else
\fi
\numberwithin{equation}{section}

\title{Design-Based Inference under \\ Random Potential Outcomes}
\ifanonymous
\author{Submitted for review}
\date{}
\else
\author{Yukai Yang\orcidlink{0000-0002-2623-8549} \\ Department of Statistics, Uppsala University}
\date{}
\fi

\begin{document}

\maketitle

\begin{abstract}
We study whether mechanism-level causal estimands, defined as expectations over latent stochastic environments, can be consistently recovered from a single realised randomised experiment. Identification alone does not guarantee recoverability. The target estimand averages over latent environments, whereas a single experiment provides only one realisation of such an environment. We show that suitably sparse local dependence induces an ergodic-type property under which cross-sectional averaging consistently recovers expectations over the latent outcome-generating mechanism. Under this structure, aggregate design-based estimators are consistent and asymptotically normal, and the variance becomes consistently estimable from a single experiment. Unlike classical finite-population inference, where Neyman-type variance estimators are structurally limited to conservative upper bounds, the proposed framework permits consistent variance estimation through the shift from fixed potential outcome schedules to stochastic mechanisms.
\end{abstract}


\section{Introduction}\label{sec:introduction}

\defcitealias{neyman1923}{Neyman (1923/1990)}

In randomised experiments, researchers are typically interested in whether an observed treatment effect reflects a stable underlying mechanism rather than a favourable realisation of a particular experimental environment. This question arises naturally in applications such as randomised controlled trials, A/B testing, industrial experimentation, and policy evaluation, where the objective is not merely to describe one realised experiment but to understand the broader mechanism generating such effects across hypothetical repetitions of the experiment. The classical fixed potential outcome (FPO) framework, however, provides no mechanism for addressing this distinction.

In the classical design-based framework of \citetalias{neyman1923} and \citet{rubin1974}, potential outcomes are treated as fixed quantities, and randomness arises solely through the treatment assignment mechanism. Under this FPO formulation, the canonical estimand is
\begin{equation}\label{eq:simpleFPOtau}
\tau_{\mathrm{FPO}}
=
\frac{1}{n}\sum_{i=1}^n
\bigl\{
Y_i(1)-Y_i(0)
\bigr\},
\end{equation}
where the realised potential outcome schedule is regarded as fixed throughout the experiment. The framework therefore conditions on a particular realised experimental world and does not distinguish between inference conditional on that realised world and recoverability of a mechanism-level quantity defined across stochastic environments. Design-based identification is achieved directly through the randomisation design without requiring a probabilistic model for the outcome process.

This framework has been extended to settings involving interference, network dependence, and complex experimental designs \citep{aronow2017estimating,savje2021average,athey2021design}. These developments preserve the FPO assumption and continue to define causal estimands conditionally on realised potential outcome schedules. Consequently, the object of inference remains tied to a realised experimental environment rather than to an underlying stochastic mechanism.

In this paper, we consider a random potential outcome (RPO) formulation. The potential outcome for unit $i$ is represented by $\tilde y_i(z,\omega)$, which depends jointly on the treatment assignment $z$ and a latent stochastic environment $\omega$.
A canonical expectation-based estimand under this formulation is
\begin{equation}\label{eq:simpleRPOtau}
\tau_{\mathrm{RPO}}
=
\mathbb E_{\omega}
\left[
\frac{1}{n}\sum_{i=1}^n
\bigl\{
\tilde y_i(1,\omega)
-
\tilde y_i(0,\omega)
\bigr\}
\right].
\end{equation}
More generally, the paper studies mechanism-level treatment-effect functionals defined through expectations over latent stochastic environments rather than conditionally on realised potential outcome schedules. This differs from model-assisted approaches \citep{rubin1978bayesian,rosenbaum1983central,chernozhukov_doubledebiased_2018}, which also allow stochastic outcome structures but rely on modelling assumptions on the outcome or assignment mechanism rather than on the randomisation design itself. The classical FPO framework is recovered as the degenerate special case in which the latent environment reduces to a singleton. In this degenerate case, the recoverability problem disappears because no averaging over latent stochastic environments is required.

It is useful to distinguish the present framework from sampling-based or superpopulation formulations \citep{cassel1977,little2004calibrated,abadie2020}. Those approaches study population-level quantities under sampling uncertainty, where randomness arises through the sampling mechanism. By contrast, the estimands considered here are expectations over latent stochastic environments and therefore represent mechanism-level quantities rather than population-level quantities. Consequently, the inferential objective differs fundamentally from that of sampling-based formulations.
In particular, the local-dependence and neighbourhood-growth conditions developed in Section~\ref{sec:asymptotic} address recoverability of expectations over latent stochastic environments, a problem that is not characterised within the sampling-based framework.

Once causal estimands are defined as expectations over latent stochastic environments, a new inferential problem emerges. Although the estimand is defined through averaging over $\omega$, only one realised experiment corresponding to one draw of $\omega$ is observed. Identification alone therefore does not guarantee recoverability from a single realised experiment. The central inferential question becomes whether cross-sectional averaging across units can asymptotically recover expectations over the latent stochastic environment.

Recoverability fundamentally depends on the dependence structure induced by the latent environment. If all units share a single common latent shock, increasing the sample size provides no additional information about expectations over the stochastic environment. Cross-sectional averaging cannot recover mechanism-level quantities because the experiment still contains only one realised draw of $\omega$.
Additional probabilistic structure is therefore required to connect cross-sectional aggregation to averaging over latent stochastic environments.
In the design-based experimental settings considered here, such structure often arises naturally from the experimental design itself through clustering, spatial organisation, controlled interaction patterns, or partially engineered dependence structures.

We show that suitably sparse local dependence provides such a bridge. Under the conditions developed below, local dependence induces an ergodic-type property under which cross-sectional averaging asymptotically recovers expectations over the latent stochastic environment. A single realised experiment therefore becomes sufficient for asymptotically valid inference on mechanism-level causal estimands. The resulting asymptotic analysis is conducted on the joint probability space generated by treatment assignment and latent outcome variation.

In canonical discrete-treatment settings, the resulting estimator coincides algebraically with familiar unbiased estimators from the FPO framework, including the Horvitz--Thompson estimator \citep{horvitz_generalization_1952}. Under the stochastic framework developed here, however, the same statistic targets the mechanism-level estimand $\tau_{\mathrm{RPO}}$ and satisfies asymptotic normality under local dependence.

A further contribution of this paper concerns variance estimation. In the classical FPO framework, Neyman-type variance estimators are structurally limited to conservative upper bounds rather than consistent estimators of the variance itself. This limitation is not a technical artefact but a consequence of finite-population reasoning under FPO schedules. Within the proposed stochastic framework, this structural limitation can be overcome. By exploiting the averaging structure induced by local dependence, we develop feasible variance estimators that are consistent under stochastic outcome variation and account jointly for randomisation and latent outcome stochasticity. This permits asymptotically valid inference from a single realised experiment. The resulting improvement arises from shifting the inferential target from fixed realised schedules to stochastic outcome-generating mechanisms rather than from technical refinement within the finite-population framework itself.

Treatment-effect functionals admitting continuous linear representations remain identifiable through the randomisation design under the proposed framework.
Our representer-based construction draws on the use of representers for identification in design-based settings developed by \citet{harshaw2022riesz}.
Their framework studies identification of treatment-effect functionals under FPO.
By contrast, the present paper studies recoverability and asymptotic inference for expectation-based causal estimands under random potential outcomes.
The central question is not whether a treatment-effect functional is identified, but whether a mechanism-level estimand defined through expectations over latent stochastic environments can be consistently recovered from a single realised experiment.
This recoverability problem is specific to estimands defined through averaging over latent stochastic environments and therefore has no direct counterpart within the FPO setting, where inference is conditional on a realised potential-outcome schedule.
In this sense, representer constructions serve as mathematical infrastructure for identification rather than as the primary inferential contribution of the paper.
Accordingly, the contribution of the present paper lies not in a new representer construction, but in establishing recoverability, asymptotic normality, and consistent variance estimation for mechanism-level causal estimands under stochastic outcome-generating environments.

The remainder of the paper is organised as follows.
Section~\ref{sec:setup} develops the RPO framework, introduces the formal assumptions, and establishes identification of the estimand.
Section~\ref{sec:asymptotic} studies consistent estimation and asymptotic normality under local dependence.
Section~\ref{sec:variance} develops feasible variance estimation.
Section~\ref{sec:simulation} presents simulation evidence.
Section~\ref{sec:conclusion} concludes.

\section{Random Potential Outcomes and Mechanism-Level Estimands}\label{sec:setup}

We begin by considering a setting with $n$ units, where
\begin{assumption}[The Stochastic Setting]\label{as:setting}
each unit's potential outcome admits a representation as a general measurable mapping
$y_i(z, x_i, \epsilon_i)$
with the following arguments: $z=(z_1,\dots,z_n)$ denoting the vector of treatment assignments;
$x_i=x_i(\omega)$ representing the possibly observed covariates for unit $i$;
$\epsilon_i=\epsilon_i(\omega)$ denoting the idiosyncratic error or unobserved heterogeneity for unit $i$.
Here, $\omega \in \Omega$ is a latent random element defined on a common probability space $(\Omega,\mathcal{F}_\omega,P)$ shared by all units.
Furthermore, we define a measurable mapping from the latent space $\Omega$ to variables $(x_i, \epsilon_i)$
and represent the potential outcome via the composition
\begin{equation}\label{eq:potential-outcome}
    \tilde{y}_i(z,\omega) = y_i(z, x_i(\omega), \epsilon_i(\omega)),
\end{equation}
since $x_i$ and $\epsilon_i$ are measurable functions of $\omega$.
\end{assumption}
Assumption~\ref{as:setting} extends the classical FPO framework by modelling potential outcomes as random elements indexed by a latent stochastic environment $\omega$.
The resulting framework defines mechanism-level treatment-effect functionals on a model space indexed jointly by treatment assignments and latent environments.
Although $\omega$ is not unit-specific, its structure may include subcomponents that affect different units differently, thereby allowing for unit-level heterogeneity and flexible dependence structures in the potential outcomes.

The potential outcome for unit $i$ depends on the entire treatment vector $z$, so that interference (or spillover effects) is permitted. Notably, one may also consider the special case of no spillover effects, whereby
$y_i(z, x_i, \epsilon_i)=y_i(z_i, x_i, \epsilon_i)$,
which corresponds to a version of the Stable Unit Treatment Value Assumption (SUTVA) incorporating consistency, as introduced by \cite{rubin1980comment}.
In the present work we do not impose the no-spillover condition a priori; rather, we allow for general interference and accommodate its presence within the proposed framework.

In practice, it often suffices to allow heterogeneity across units to arise through $(x_i, \epsilon_i)$, making it unnecessary to index the structural function by $i$.
In such cases, one may work with a common measurable function $y$, while still inducing heterogeneous potential outcomes via the mappings $\tilde{y}_i(z, \omega) = y(z, x_i, \epsilon_i)$.

The structural mapping \eqref{eq:potential-outcome} is the primary object of interest in what follows. Following randomisation principles, we assume that
\begin{assumption}[Randomisation]\label{as:randomisation}
the treatment assignment vector $z$ is drawn from a known randomisation distribution that is independent of the outcome-generating process given the latent variable $\omega$.
\end{assumption}
Assumption~\ref{as:randomisation} ensures that the design distribution remains valid conditional on $\omega$, which permits identification through the randomisation design.

The framework permits both randomised and observational assignment mechanisms, depending on the dependence structure between $z$ and $\omega$.
In the randomised setting considered throughout this paper, the design ensures independence between treatment assignment and latent outcome variation.
In all cases, the framework models potential outcomes as random rather than fixed.

To define the model space used below, we first introduce the corresponding function spaces.
Let $L^p(\mathcal{Z} \times \Omega)$ denote the space of all measurable functions $u: \mathcal{Z} \times \Omega \to \mathbb{R}$,
satisfying $\mathbb{E}[|u(z,\omega)|^p]<\infty$, for some integer $p \geq 1$.
For notational convenience, we henceforth write $L^p$ in place of $L^p(\mathcal{Z} \times \Omega)$,
whenever the meaning is clear from context.
Furthermore, we denote the $p$-norm of $u \in L^p$ by
$\|u\|_p := \left( \int_{\mathcal{Z}\times\Omega} \left| u(z, \omega) \right|^p \mu(\dif z) P(\dif \omega) \right)^{1/p} = \left( \mathbb{E} \left| u(z, \omega) \right|^p \right)^{1/p}$,
where $\mu$ denotes the probability measure for randomisation.
Under Assumption~\ref{as:randomisation}, the joint law of $(z,\omega)$ factorises as $\mu(\dif z) P(\dif \omega)$, which permits rewriting the $L^p$-norm as an expectation.
We write $\|u\| := \|u\|_2$ for brevity if $u \in L^2$.

We model each unit's potential outcome function, via the structural mapping \eqref{eq:potential-outcome}, as an element of a model space $\mathcal{M}_i$,
which is a subspace of $L^2$.
Formally, we write
\begin{assumption}[Model Space]\label{as:mspace}
$\tilde{y}_i(z, \omega) \in \mathcal{M}_i \subset L^2(\mathcal{Z} \times \Omega)$.
\end{assumption}
The model space $\mathcal M_i$ incorporates both assignment randomness and outcome-level randomness through the joint variable $(z,\omega)$.

The FPO assumption can be viewed as a special case of our model,
corresponding to the degenerate setting where $\Omega = \{\omega_0\}$ for some $\omega_0$.
In this case, the latent variable is effectively fixed, and the potential outcome function
becomes a deterministic function of $z$, recovering the FPO framework.

\comm{about the inner product}

We equip $\mathcal{M}_i$ with the inner product
\begin{equation} \label{eq:inner-product}
\langle u,v \rangle := \mathbb{E}_\omega\Bigl[\mathbb{E}_{z}\bigl[u(z,\omega)\,v(z,\omega)\bigr]\Bigr] 
=\mathbb{E}\bigl[u(z,\omega)\,v(z,\omega)\bigr],
\end{equation}
where the second equality follows from the law of iterated expectations.
The associated norm is given by $\|u\| = \sqrt{\langle u, u \rangle}$, which coincides with the standard $L^2$-norm.
In the degenerate case where $\Omega = \{\omega_0\}$, the inner product reduces to
\begin{equation} \label{eq:inner-product0}
\langle u,v \rangle^0 = \mathbb{E}_{z}\bigl[u(z,\omega_0)\,v(z,\omega_0)\bigr],
\end{equation}
which coincides with the inner product used by \citet{harshaw2022riesz}.

Whenever necessary, we implicitly work with the completion of $\mathcal{M}_i$ under the inner product defined in \eqref{eq:inner-product}.
Let $\mathcal{M}_i'$ denote the space of continuous linear functionals on $\mathcal{M}_i$.
\begin{assumption}[Dual Representability of the Treatment Effect] \label{as:dual}
The treatment effect $\theta_i$ for each unit $i$ can be represented as a continuous linear functional on the model space of potential outcomes, \emph{viz.},
$\theta_i : \mathcal{M}_i \to \mathbb{R}$ with $\theta_i \in \mathcal{M}_i'$.
\end{assumption}
This assumption formalises treatment effects as continuous linear functionals on the model space of potential outcomes.
Such a formulation yields an identification construction in which the treatment effect is represented through an inner product with a representer function.
The same construction applies across both FPO and RPO settings.

Assumption~\ref{as:dual} excludes pointwise evaluation maps. In an $L^2$ model space, evaluation at a single point is not well-defined on equivalence classes and, even when viewed on a representative function space, is not continuous with respect to the $L^2$ norm.

In all cases, $\theta_i$ is a predetermined linear functional on $\mathcal{M}_i$, fixed throughout the analysis and not learned from data.
We illustrate admissible choices of $\theta_i$ with the following examples of continuous linear functionals.

\begin{example}{Expected treatment effect for binary treatment}
\begin{equation}\label{eq:example1}
\theta_i(u)
\;=\;
\mathbb{E}_\omega\bigl[u(z^{(1)},\omega) - u(z^{(0)},\omega)\bigr]
\end{equation}
representing the expected outcome difference from shifting treatment assignment from \(z^{(0)}\) to \(z^{(1)}\).
\end{example}

\begin{example}{Expected treatment effect for treatment intervals or regions}
\begin{equation}\label{eq:example2}
\theta_i(u)
\;=\;
\mathbb{E}_\omega\Bigl[\mathbb{E}_{z \in A}\bigl[u(z,\omega)\bigr]
- \mathbb{E}_{z \in B}\bigl[u(z,\omega)\bigr]\Bigr],
\end{equation}
where $A, B \subset \mathcal{Z}$ are measurable sets with $A \cap B = \emptyset$,
representing the expected outcome difference between treatment groups $A$ and $B$.
\end{example}

\begin{example}{Expected partial derivative}
\begin{equation}\label{eq:example3}
\theta_i(u) = \mathbb{E}_\omega \! \left[ \partial_z u(z_0, \omega) \right],
\end{equation}
at some point $z_0 \in \mathcal{Z}$,
representing the expected marginal effect at $z_0$.
\end{example}

\begin{example}{Test function weighted expected partial derivative}
\begin{equation}\label{eq:example4}
\theta_i(u)
=\mathbb E_\omega\!\left[\int_{\mathcal Z}\partial_z u(z,\omega)\,\varphi(z)\,\dif z\right],
\end{equation}
where $\varphi$ is a predetermined test function that weights the partial derivative over
$z\in\mathcal Z$.
In particular, taking $\varphi=\mathbf 1_A$ for a measurable set $A\subset\mathcal Z$ yields
\begin{equation}\label{eq:exam4special}
\theta_i(u)
=\mathbb E_\omega\!\left[\int_{A}\partial_z u(z,\omega)\,\dif z\right],
\end{equation}
representing the expected marginal effect averaged over region $A$.
\end{example}

The following propositions establish sufficient conditions under which the preceding examples define linear and continuous functionals.
\begin{proposition}\label{pp:example12}
Suppose that $\mu(A)>0$ and $\mu(B) > 0$. The functional \eqref{eq:example2} is linear and continuous on $\mathcal{M}_i$.
\end{proposition}
Note that Example~\eqref{eq:example1} is a special case of Example~\eqref{eq:example2}, obtained by setting $A = \{z^{(1)}\}$ and $B = \{z^{(0)}\}$.
The condition $\mu(A) > 0$ and $\mu(B) > 0$ is commonly referred to as the \emph{positivity} or \emph{overlap} assumption in the causal inference literature.
It ensures that both treatment groups are represented in the randomisation distribution, which guarantees that the corresponding estimand is well-defined and estimable.
In our case, the positivity assumption plays a crucial role in verifying that the functional is both linear and continuous.

\begin{proposition}\label{pp:point_derivative}
Let $\mathcal Z\subset\mathbb R^d$ be open and bounded, fix $z_0\in\mathcal Z$, and let $s>1+\frac d2$.
Assume
$u\in L^2(\Omega;H^s(\mathcal Z))$\footnote{$L^2(\Omega; H^{s}(\mathcal{Z}))$ denotes the Bochner space with values in $H^{s}(\mathcal{Z})$, and $H^{s}(\mathcal{Z}) = W^{s,2}(\mathcal{Z})$ is the Sobolev space of order $s$.}.
Then the functional \eqref{eq:example3} is linear and continuous on $L^2(\Omega;H^s(\mathcal Z))$.
\end{proposition}

This result implies that, for the functional to satisfy the necessary continuity condition,
the model space must be further restricted to
$\mathcal{M}_i \subset L^2(\Omega; H^s(\mathcal{Z})) \subset L^2(\mathcal Z\times\Omega)$.
Within this space, functions admit $C^1$ representatives in the $z$-variable almost surely,
so that pointwise evaluation of derivatives is well-defined and continuous.

\begin{proposition}\label{pp:region_integral}
Let $\mathcal Z\subset\mathbb R^d$ be open and bounded and let $\varphi\in L^2(\mathcal Z)$.
Assume $u\in L^2(\Omega;H^1(\mathcal Z))$.
Then the functional \eqref{eq:example4} is linear and continuous on $L^2(\Omega;H^1(\mathcal Z))$.
\end{proposition}

This proposition is deliberately permissive in that the test function $\varphi$ is only required to belong to $L^2(\mathcal Z)$ and is therefore not restricted to indicator functions such as $\mathbf 1_A$ in \eqref{eq:exam4special}. More general weights, including discontinuous or piecewise-defined functions, are admissible. In contrast, Dirac-type test functions corresponding to point evaluation are not elements of $L^2(\mathcal Z)$ and hence fall outside the scope of this result. Such pointwise functionals require additional regularity and are instead covered by Proposition~\ref{pp:point_derivative}.

These regularity conditions ensure that the corresponding treatment-effect functionals remain well-defined and continuous on the model space.

\comm{Riesz representation theorem}

\begin{theorem}[Riesz Representation Theorem in \(\mathcal{M}_i\)] \label{th:Riesz}
Under Assumptions~\ref{as:setting}, \ref{as:mspace}, and \ref{as:dual},  
for every \( \theta_i \in \mathcal{M}_i' \),  
there exists a unique element \( \psi_i \in \mathcal{M}_i \) such that
$\theta_i(u) = \langle u, \psi_i \rangle$ for all $u \in \mathcal{M}_i$.
\end{theorem}
Theorem~\ref{th:Riesz} establishes identification of the treatment-effect functional through the corresponding \emph{Riesz representer} $\psi_i$.
This establishes identification of the mechanism-level estimand on the model space.
Whether the estimand can be consistently recovered from a single realised experiment is a separate inferential question, addressed in Section~\ref{sec:asymptotic}.

\comm{Riesz estimator}

By constructing the Riesz representer \(\psi_i\) corresponding to the treatment effect functional \(\theta_i\), we only consider those treatment effects that are representable in the dual space \(\mathcal{M}_i'\).
Consequently, we estimate the treatment effect
\begin{equation}\label{eq:treatment-for-y}
    \theta_i(\tilde{y}_i) = \langle \tilde{y}_i, \psi_i \rangle = \mathbb{E}\Bigl[\tilde{y}_i(z,\omega)\,\psi_i(z,\omega)\Bigr],
\end{equation}
by using the corresponding \emph{Riesz estimator}
\begin{equation}\label{eq:Riesz-estimator-i}
    \hat{\theta}_i(z, \omega) = \tilde{y}_i(z,\omega)\, \psi_i(z, \omega),
\end{equation}
as, in practice, we observe only one realisation of the outcome \(\tilde{y}_i(z,\omega)\) for certain $z$ and $\omega$.
It follows directly from the linearity of the Riesz representation and the joint expectation over $(z,\omega)$ that the estimator in \eqref{eq:Riesz-estimator-i} is unbiased for \eqref{eq:treatment-for-y}.
This unbiasedness holds with respect to the combined randomness of the assignment mechanism and the outcome-generating environment, and does not rely on conditioning on a realised potential outcome schedule.

In practice, the Riesz representer can be approximated using standard basis expansions and finite-dimensional projections under the known design; see Appendix~\ref{sec:estimation} for computational details.

\comm{the estimand and unbiasedness of the Riesz estimator}

\begin{definition}[Aggregate Estimand]\label{df:estimand}
The aggregate estimand considered in this paper is defined as a weighted average of individual treatment effects
\begin{equation}\label{eq:finite-estimand}
\tau_n = \sum_{i=1}^n \nu_{ni} \theta_i(\tilde{y}_i),
\end{equation}
where the weights \( \nu_{ni} \in \mathbb{R} \) satisfy \( \nu_{ni} \geq 0 \).
\end{definition}
The canonical RPO estimand $\tau_{\mathrm{RPO}}$ in \eqref{eq:simpleRPOtau} arises as a special case under particular choices of the weights.
More generally, the formulation accommodates arbitrary fixed weighting schemes $\nu_{ni}$.

The estimand $\tau_n$ is defined on the joint law of $(z,\omega)$ and remains mathematically well-defined regardless of whether treatment assignment and latent outcome variation are independent.
Different assumptions on the joint law therefore induce different causal or associational interpretations.

The classical FPO formulation is recovered as a degenerate special case,
\begin{equation}\label{eq:finite-estimand-harshaw}
\tau_n^0
\;=\;
\frac{1}{n} \sum_{i=1}^n \langle \tilde{y}_i, \psi_i \rangle^0
\;=\;
\frac{1}{n} \sum_{i=1}^n \mathbb{E}_z\Bigl[\tilde{y}_i(z,\omega_0)\,\psi_i(z,\omega_0)\Bigr],
\end{equation}
where the latent environment $\omega_0$ is degenerate and the treatment effect is deterministic.
In this setting, the average is equally weighted across units.
The canonical two-treatment average treatment effect $\tau_{\mathrm{FPO}}$ in \eqref{eq:simpleFPOtau} arises as a further special case under uniform weights.

\comm{weights}

A common choice for the weights is the uniform weighting $\nu_{ni} = 1/n$.
More generally, this formulation accommodates subgroup-specific or covariate-adjusted averages, depending on the structure of the weights.
We impose the following assumption to control the magnitude of the weights
\(\nu_{ni}\).
\begin{assumption}[Uniform Weight Upper Bound]\label{as:weights}
There exists a constant \( \bar{\nu} > 0 \) such that
$n \nu_{ni} \le \bar{\nu}$ for all $i$ and all $n \in \mathbb{N}$.
\end{assumption}
It guarantees that each weight is bounded by a constant multiple of \(1/n\),
uniformly in \(i\), without requiring the individual sequences
\(\{n\nu_{ni}\}_{n\ge i}\) to converge.
Consequently, no finite subset of indices can receive an asymptotically
dominant share of the total weight.
Note that Assumption~\ref{as:weights} implies $\sup_{n\to\infty} \sum_{i=1}^n \nu_{ni} < \infty$ but not necessarily one, which is quite general.
In the continuous analogue this corresponds to the integrability of the weight function.

Given the unit-level Riesz estimators \eqref{eq:Riesz-estimator-i} and the weights,
\begin{definition}[Aggregate Riesz Estimator]
we define the corresponding aggregate Riesz estimator for the finite-sample estimand 
\begin{equation}\label{eq:finite-estimator}
\hat{\tau}_n(z, \omega) = \sum_{i=1}^n \nu_{ni} \hat{\theta}_i(z, \omega).
\end{equation}
\end{definition}
Thus we have the unbiasedness of the Riesz estimator in the following.

\begin{theorem}[Unbiasedness of the Riesz Estimator]\label{th:unbiasedness}
The Riesz estimator defined in \eqref{eq:Riesz-estimator-i} is unbiased for the individual treatment effect,
\emph{i.e.},
$\mathbb{E}\bigl[\hat{\theta}_i(z, \omega)\bigr] = \theta_i(\tilde{y}_i)$.
Consequently, the aggregate Riesz estimator in \eqref{eq:finite-estimator} is unbiased for the finite-sample estimand \eqref{eq:finite-estimand}, \emph{i.e.},
$\mathbb{E}\Bigl[\hat{\tau}_n(z, \omega)\Bigr] = \tau_n$.
\end{theorem}
In the stochastic setting, the Riesz representer $\psi_i$ depends on the latent environment $\omega$ and is therefore itself random.
In any realised experiment, estimation is necessarily conditional on a single realised environment.
The asymptotic theory developed below accounts jointly for randomness in both $\tilde y_i$ and $\psi_i$.

The preceding results establish identification and unbiasedness on the model space.
However, these properties alone do not guarantee consistent recovery of the estimand from a single realised experiment.
Additional probabilistic structure is required to connect cross-sectional aggregation to averaging over the latent stochastic environment.
The next section develops conditions under which such recoverability holds.

\section{Recoverability under Local Dependence}\label{sec:asymptotic}

Although the estimand is identified through the randomisation design, recoverability from a single realised experiment requires additional probabilistic structure linking cross-sectional averaging to averaging over latent stochastic environments.
The results below give a boundary between identifiable but unrecoverable estimands and identifiable and recoverable estimands.

The local dependence conditions introduced below induce an ergodic-type averaging structure under which cross-sectional aggregation asymptotically recovers the expectation-based estimand from a single realised experiment.
These conditions should therefore be interpreted as inferential regularity conditions for recoverability rather than as structural assumptions on treatment interactions or causal mechanisms.

Cross-sectional dependence is modelled through local dependency neighbourhoods following \citet{chen2004normal}.
The local dependence structure is imposed only on the observable pairs $(\tilde y_i,\psi_i)$ appearing in the estimator construction, rather than directly on the latent stochastic environment itself.

\begin{definition}[Dependency Neighbourhoods]\label{df:dep-neighbour}
For each unit \( i \), let \( \{ M_{ik} \}_{k \in I_i} \) be the collection of all subsets \( M_{ik} \subset \{1, 2, \dots, n\} \) such that the pair of random variables \( (\tilde{y}_i, \psi_i) \) is independent of the set
$\left\{ (\tilde{y}_j, \psi_j) : j \notin M_{ik} \right\}$.
The \emph{dependency neighbourhood} of unit \( i \) is defined as the intersection over all such sets,
$N_i := \bigcap_{k \in I_i} M_{ik}$.
\end{definition}

This definition follows the dependency-neighbourhood formulation of \citet[Definition 3.1]{ross2011fundamentals}.
We denote by \( |N_i| \) the number of units in the dependency neighbourhood of unit \( i \), referred to as the neighbourhood size for unit \( i \).
Let \( D_n = \max_{1 \leq i \leq n} |N_i| \) and \( d_n = n^{-1} \sum_{i=1}^{n} |N_i| \) denote, respectively, the maximum and average neighbourhood sizes across the sample.
These quantities play a central role in the asymptotic analysis below.

In practice, the specification of the dependency neighbourhoods \(N_i\) relies on substantive knowledge about the qualitative structure of dependence in the experimental environment rather than direct observation of the latent stochastic environment itself.
In many design-based experimental settings, such dependence structures are not arbitrary features inferred post hoc from observational data, but arise directly from the experimental design through clustering, spatial allocation, temporal ordering, production architecture, or controlled interaction structures.
Examples include classroom-based interventions, village-level randomisation, spatial field experiments, factory or production-line experiments, and large-scale A/B testing systems with partially known infrastructure interactions.

The local dependence assumptions considered here should therefore be interpreted as inferential regularity conditions induced or justified by the experimental design itself, rather than as unrestricted assumptions on a generic observational dependence structure.
The framework does not require modelling or observing the latent environment $\omega$ directly.
Instead, it formalises when a designed experiment contains sufficient distributed information across units to asymptotically recover mechanism-level quantities from a single realised experiment.


\begin{assumption}[Maximum Dependency Neighbourhood Size]\label{as:mns}
The maximum dependency neighbourhood size $D_n$ satisfies
$D_n = o(n^{d})$
for some $d > 0$.
\end{assumption}

\begin{assumption}[Average Dependency Neighbourhood Size]\label{as:ans}
The average dependency neighbourhood size $d_n$ satisfies either (a) $d_n = o(n)$ or (b) $\sup_{n \in \mathbb{N}} d_n < \infty$.
\end{assumption}
These assumptions regulate the growth of the dependency neighbourhoods and constitute the sparsity conditions used throughout the asymptotic analysis.

\comm{max-$p$ norm}

In addition, for a sequence of functions $u_1, \dots, u_n \in L^p$, where $p > 1$, we define the max-$p$ norm as
$\norm{u}_{\max,p}^n := \max_{1 \leq i \leq n} \norm{u_i}_p$.
The sequence $\norm{u}_{\max, p}^n$ is clearly non-decreasing in $n$,
and, as such, is bounded if and only if it converges as $n \to \infty$.

\comm{consistency}

\begin{assumption}[Boundedness of Norms]\label{as:bounded-norms}
For some $r \ge 1$, there exist real numbers \(p, q \in [r, \infty]\) satisfying $1/p + 1/q = 1/r$ such that
$\sup_{i\in\mathbb{N}} \norm{\tilde{y}_i}_{p} < \infty$ and 
$\sup_{i\in\mathbb{N}} \norm{\psi_i}_{q}  < \infty$.
\end{assumption}
This assumption allows flexibility in the choice of \( r \).
It requires uniform boundedness of the \( p \)th moments of \( \tilde{y}_i \) and the \( q \)th moments of \( \psi_i \), thereby ruling out heavy-tailed behaviour.

The following theorem establishes consistency of the aggregate Riesz estimator.
\begin{theorem}[Consistency and Recoverability]\label{th:consistency}
Under Assumptions~\ref{as:setting}, \ref{as:mspace}--\ref{as:weights}, \ref{as:ans}$(a)$, and \ref{as:bounded-norms}$_{(r=2)}$,
the aggregate Riesz estimator in \eqref{eq:finite-estimator} is consistent in mean square for $\tau_n$.
Moreover, if Assumption~\ref{as:ans}(b) holds, then
$\hat{\tau}_n - \tau_n = O_p(n^{-1/2})$.
\end{theorem}
Consistency here means that cross-sectional averaging asymptotically recovers expectations over latent stochastic environments from a single realised experiment.
This differs fundamentally from the classical FPO framework, where the estimand is conditioned on the realised environment itself.
Without suitable restrictions on the dependence structure, recoverability may fail even asymptotically.
For example, if all units share a common latent shock, increasing the sample size does not increase information about expectations over the latent stochastic environment.

Let $\sigma_n^2 := \mathbb{V}[\hat{\tau}_n(z,\omega)]$ denote the variance of the aggregate Riesz estimator.
Consistency alone does not imply asymptotic normality, so additional conditions are required.
The following non-degeneracy assumption imposes a lower bound on the variance of the aggregate Riesz estimator and plays a critical role in ensuring asymptotic normality.
When combined with the previous assumptions, it provides sufficient conditions for the aggregate Riesz estimator to satisfy a central limit theorem.
\begin{assumption}[Variance Lower Bound]\label{as:lowerbound-sigma}
There exists \( \sigma_0 > 0 \) such that \( \inf_{n \in \mathbb{N}} \sqrt{n} \sigma_n \geq \sigma_0 \).
\end{assumption}
This condition ensures that the asymptotic variance does not degenerate under root-$n$ scaling.

We require further the following assumption to establish the asymptotic normality result,
which concerns the existence of finite fourth moments.
\begin{assumption}[Finite Fourth Moment]\label{as:ffm}
The individual treatment effect estimator satisfies either
(a) $\mathbb{E}\left[\hat{\theta}_i(z, \omega)^4\right] < \infty$
or (b) $\sup_{i \in \mathbb{N}} \mathbb{E}\left[\hat{\theta}_i(z, \omega)^4\right] < \infty$.
\end{assumption}

The following theorem establishes asymptotic normality of the aggregate Riesz estimator.
\begin{theorem}[Asymptotic Normality]\label{th:normality}
Under Assumptions~\ref{as:setting}, \ref{as:mspace}--\ref{as:weights}, \ref{as:mns}\(_{(d=1/4)}\), \ref{as:bounded-norms}\(_{(r=3)}\), \ref{as:bounded-norms}\(_{(r=4)}\), \ref{as:lowerbound-sigma}, and \ref{as:ffm}$(a)$,
the aggregate Riesz estimator defined in \eqref{eq:finite-estimator} satisfies
$\sigma_n^{-1} \left( \hat{\tau}_n(z, \omega) - \tau_n \right) \xrightarrow{d} \mathcal{N}(0, 1)$.
\end{theorem}
Both Assumptions~\ref{as:bounded-norms}\(_{(r=3)}\) and~\ref{as:bounded-norms}\(_{(r=4)}\) are required here, as
boundedness at $r = 4$ does not imply the same at $r = 3$ over $\mathcal{Z} \times \Omega$, due to the potential unboundedness of the joint space.

Theorem~\ref{th:normality} justifies asymptotically valid inference based on the standardised statistic
$T_n = \hat{\tau}_n(z,\omega)/\sigma_n$.
The sparsity restriction therefore reflects a trade-off between dependence complexity and recoverability of mechanism-level information from a single realised experiment.

\section{Variance Estimation under Local Dependence}\label{sec:variance}

Existing design-based approaches under interference often rely on conservative variance bounds rather than consistent estimation of the exact sampling variance.
This limitation is structural in the classical FPO framework, where inference conditions on a realised schedule of potential outcomes rather than averaging over latent outcome-generating environments.

Under the present framework, the inferential target is instead defined over the joint law of the treatment assignment and latent stochastic environment.
The variance therefore becomes a population-level quantity associated with the stochastic mechanism itself rather than with a fixed realised schedule.
In this sense, the proposed framework changes the inferential structure of design-based inference.

Recoverability of the variance from a single realised experiment nevertheless remains nontrivial.
Consistent estimation becomes possible only because the local dependence structure introduced in Section~\ref{sec:asymptotic} induces an ergodic-type averaging mechanism across units.
The improvement from conservative variance bounds to consistent variance estimation is therefore not merely a technical refinement, but a consequence of changing the inferential target from fixed potential outcome schedules to stochastic mechanisms.

Existing approaches under interference typically focus on conservative variance control under structured dependence.
For example, \citet{aronow2017estimating} rely on known exposure mappings and strictly positive joint inclusion probabilities; \citet{yu_estimating_2022} assume bounded-degree network structures; and \citet{liu2014large} consider two-stage designs with bounded group sizes.
\citet{harshaw2022riesz} construct an unbiased estimator for an upper variance bound using a tensor-product Riesz representation, while \citet{harshaw2024optimizing} develop optimisation procedures for tightening such conservative bounds across experimental designs.

In the present framework, dependence arises through the latent stochastic environment $\omega$ and is reflected through the observable pairs $(\tilde y_i,\psi_i)$.
Although the latent environment itself is unobserved, qualitative features of the dependence structure are often informed by the experimental design or application context.
In many settings, the design itself partially determines which units are likely to exhibit meaningful interaction or dependence.

This section develops feasible variance estimators for the aggregate Riesz estimator under local dependence.
The variance structure depends explicitly on the dependency neighbourhoods introduced in Section~\ref{sec:asymptotic}, and consistent estimation of the exact sampling variance therefore requires structural information about the underlying dependence pattern.
The proposed estimator exploits sparsity of the dependency neighbourhood structure to consistently recover the sampling variance under suitable neighbourhood growth conditions.

The variance estimator is directly implementable from observed data once the representer $\psi_i(\cdot,\omega_0)$ has been approximated.
Given the observed realised environment $\omega_0$, the quantity $\hat\theta_i(z,\omega_0)$ is computable from \eqref{eq:Riesz-estimator-i}.
In practice, $\psi_i(\cdot,\omega_0)$ may be approximated using standard finite-dimensional projections; see Appendix~\ref{sec:estimation} for computational details.

Under sharp null hypotheses specifying $\theta_i(\tilde y_i)$, define the centred quantity
\begin{equation}
\zeta_i(z, \omega)
=
\hat{\theta}_i(z,\omega) - \theta_i(\tilde y_i).
\end{equation}
Its realised value $\zeta_i(z,\omega_0)$ is then observable and can be used for variance estimation.

The proposed approach retains only those cross terms corresponding to pairs believed to exhibit dependence, while setting all remaining terms to zero.
This yields a feasible variance estimator adapted to the assumed dependency neighbourhood structure.

Let \( \vec{\zeta}_n = (\zeta_1, \dots, \zeta_n)' \) denote a vector of centred random variables with covariance matrix \( \Sigma_n = \mathbb{E}[\vec{\zeta}_n \vec{\zeta}_n'] \).
From a single realised experiment, we observe only the outer product
\( \hat{\Sigma}_n = \vec{\zeta}_n \vec{\zeta}_n' \),
rather than the population covariance matrix \( \Sigma_n \) itself.
Without additional structure, consistent estimation of the full covariance matrix is impossible.
Even under independence, the entries of \( \hat{\Sigma}_n \) typically have variance of order one, implying
\begin{equation}
\mathbb{E}\bigl[\| \hat{\Sigma}_n - \Sigma_n \|_F^2\bigr]
\ge
c n^2
\end{equation}
for some constant \( c > 0 \),
where $\norm{\cdot}_F$ denotes the Frobenius norm.
Thus the estimation error diverges with $n$.

Consistent covariance estimation therefore requires additional structure.
Examples include repeated independent experimental replications, weak stationarity combined with temporal or spatial ordering, factor structures enabling low-rank regularisation, or asymptotically vanishing second-order dependence.
Such assumptions are unavailable in the single-realisation design-based setting considered here.
The proposed approach instead targets only those covariance components associated with locally dependent pairs.

Suppose the dependency structure is given by
\begin{equation}\label{eq:dependency-structure}
\mathcal{E}_n = \{(i,j) \in \{1, \dots, n\}^2 : \zeta_i \text{ and } \zeta_j \text{ are dependent} \}.
\end{equation}
We define the local-dependence variance estimator
\begin{equation}\label{eq:plugin-vestimator}
\hat{\sigma}_n^2 = \sum_{(i,j) \in \mathcal{E}_n} \nu_{ni} \nu_{nj} \zeta_i \zeta_j,
\end{equation}
and the population variance
\begin{equation}\label{eq:population-variance}
\sigma_n^2 = \sum_{(i,j) \in \mathcal{E}_n} \nu_{ni} \nu_{nj} \mathbb{E}[\zeta_i \zeta_j].
\end{equation}
The following theorem shows that, under mild regularity conditions, \( \hat{\sigma}_n^2 \) consistently estimates \( \sigma_n^2 \).

\begin{theorem}[Consistency of the Variance Estimator]\label{th:consistency-variance}
Under Assumptions~\ref{as:setting}, \ref{as:mspace}--\ref{as:weights}, and \ref{as:ffm}$(b)$,
\( n \hat{\sigma}_n^2 - n \sigma_n^2 = O_p(n^{-1/2} D_n^{3/2})\). In particular, under Assumption~\ref{as:mns}$_{(d=1/3)}$,  $n \hat{\sigma}_n^2 - n \sigma_n^2$ converges to zero in mean square .
\end{theorem}

\begin{corollary}[Asymptotic Normality]\label{co:normality}
Under Assumptions~\ref{as:setting}, \ref{as:mspace}--\ref{as:weights}, \ref{as:mns}\(_{(d=1/4)}\), \ref{as:bounded-norms}\(_{(r=3)}\), \ref{as:bounded-norms}\(_{(r=4)}\), \ref{as:lowerbound-sigma}, and \ref{as:ffm}$(b)$,
the aggregate Riesz estimator defined in \eqref{eq:finite-estimator} satisfies 
$\hat{\sigma}_n^{-1} \left( \hat{\tau}_n(z, \omega) - \tau_n \right) \xrightarrow{d} \mathcal{N}(0, 1)$.
\end{corollary}
Corollary~\ref{co:normality} establishes that the aggregate Riesz estimator is asymptotically normal under local dependence, provided that the dependency neighbourhoods are sufficiently sparse and the variance estimator \( \hat{\sigma}_n^2 \) is properly constructed.

\comm{correlation variance estimator}

We next consider variance estimators constructed under progressively weaker information about the dependence structure.
A natural question is whether one can use a reduced dependency structure based solely on correlation.
In many experimental settings, it is difficult to explicitly define or justify a full dependency neighbourhood structure \( \mathcal{E}_n \) based on latent interference or unobserved design constraints.
However, domain knowledge or structural assumptions may suggest which unit-level estimators \( \zeta_i \) and \( \zeta_j \) are likely to be uncorrelated, even if their full dependence is unknown.
This motivates constructing a variance estimator by summing only over pairs believed to exhibit non-negligible second-order dependence, resulting in a conservative correlation-based estimator.
Specifically, consider the set
\begin{equation}\label{eq:correlation-structure}
\mathcal{E}_n^c = \{(i,j) \in \{1, \dots, n\}^2 : \mathbb{E}[\zeta_i \zeta_j] \neq 0 \},
\end{equation}
and define the corresponding correlation-based variance estimator
\begin{equation}\label{eq:plugin-vestimator2}
\hat{\sigma}_{cn}^2 = \sum_{(i,j) \in \mathcal{E}_n^c} \nu_{ni} \nu_{nj} \zeta_i \zeta_j.
\end{equation}
The next result shows that restricting the estimator to correlated pairs preserves consistency under the same sparsity conditions.
\begin{theorem}[Consistency of the Variance Estimator]\label{th:consistency-variance2}
Under Assumptions~\ref{as:setting}, \ref{as:mspace}--\ref{as:weights}, and \ref{as:ffm}$(b)$,
\( n \hat{\sigma}_{cn}^2 - n \sigma_n^2 = O_p(n^{-1/2} D_n^{3/2})\). In particular, under Assumption~\ref{as:mns}$_{(d=1/3)}$,  $n \hat{\sigma}_{cn}^2 - n \sigma_n^2$ converges to zero in mean square .
\end{theorem}

\begin{corollary}[Asymptotic Normality]\label{co:normality2}
Under Assumptions~\ref{as:setting}, \ref{as:mspace}--\ref{as:weights}, \ref{as:mns}\(_{(d=1/4)}\), \ref{as:bounded-norms}\(_{(r=3)}\), \ref{as:bounded-norms}\(_{(r=4)}\), \ref{as:lowerbound-sigma}, and \ref{as:ffm}$(b)$,
the aggregate Riesz estimator defined in \eqref{eq:finite-estimator} satisfies
$\hat{\sigma}_{cn}^{-1} \left( \hat{\tau}_n(z, \omega) - \tau_n \right) \xrightarrow{d} \mathcal{N}(0, 1)$.
\end{corollary}
Although the correlation-based estimator omits many cross-terms in the full variance expression, it remains valid provided that omitted pairs correspond to units whose second-order dependence is negligible.

\comm{Applicability to FPO Settings.}

Although the variance estimators are developed under stochastic potential outcomes, the construction also applies to the FPO framework as a degenerate special case obtained by taking $\Omega$ to be a singleton.
In that setting, randomness arises solely through the assignment mechanism.
The distinction is therefore not the algebraic form of the variance estimator itself, but the inferential interpretation of the target variance.
Under FPO, the variance is conditional on a fixed realised schedule of potential outcomes.
Under the present framework, the variance reflects uncertainty over the latent stochastic mechanism generating the potential outcomes.
The proposed framework therefore permits consistent estimation of mechanism-level sampling variability under stochastic environments.

In the remainder of the paper, we adopt the correlation-based variance estimator \( \hat{\sigma}_{cn}^2 \), while emphasising that the overall methodology still relies on the local dependency assumption, even though it is not explicitly invoked in the variance formula. Crucially, while uncorrelatedness justifies omitting second-order cross terms in the estimator, it does not imply the vanishing of higher-order joint moments across units. Therefore, a corresponding correlation-based neighbourhood structure cannot be meaningfully defined for the full dependency graph. 

In practice, this means we continue to assume a local dependency structure satisfying Assumption~\ref{as:mns} with \( d = 1/4 \), while approximating \( \sigma_n^2 \) by summing only over pairs \( (i,j) \) for which \( \zeta_i \) and \( \zeta_j \) are correlated.

\comm{in practice how to estimate the variance}

Suppose that, in practice, the correlation structure is approximated by a set of index pairs
\( \tilde{\mathcal{E}}_n \subset \{1, \dots, n\}^2 \), serving as a practical substitute for \eqref{eq:correlation-structure}. To ensure valid inference, the construction of \( \tilde{\mathcal{E}}_n \) should be conservative, satisfying
\begin{assumption}[Conservative Set of Index Pairs] \label{as:conserv}
$\mathcal{E}_n^c \subset \tilde{\mathcal{E}}_n \subset \mathcal{E}_n$.
\end{assumption}
Based on this, we define a sparse matrix \( \hat{\Sigma}_n^d \in \mathbb{R}^{n \times n} \) that retains only the entries corresponding to the index set \( \tilde{\mathcal{E}}_n \), representing the estimated dependency structure.
\begin{equation}
(\hat{\Sigma}_n^d)_{ij} \;=\; 
\begin{cases}
\zeta_i \zeta_j & \text{if } (i,j) \in \tilde{\mathcal{E}}_n, \\
0 & \text{otherwise}.
\end{cases}
\end{equation}
The corresponding variance estimator is then given by
\begin{equation}
\tilde{\sigma}_{n}^2 \;=\; \nu_n' \hat{\Sigma}_n^d \nu_n,
\end{equation}
where \( \nu_n = (\nu_{n1}, \dots, \nu_{nn})' \) is the vector of aggregation weights.

By appropriately selecting the conservative set of index pairs, the corresponding variance estimator remains consistent.
\begin{theorem}[Consistency of the Variance Estimator]\label{th:consistency-variance3}
Under Assumptions~\ref{as:setting}, \ref{as:mspace}--\ref{as:weights}, \ref{as:ffm}$(b)$, and \ref{as:conserv},
\( n \tilde{\sigma}_{n}^2 - n \sigma_n^2 = O_p(n^{-1/2} D_n^{3/2})\). In particular, under Assumption~\ref{as:mns}$_{(d=1/3)}$,  $n \tilde{\sigma}_{n}^2 - n \sigma_n^2$ converges to zero in mean square .
\end{theorem}

\begin{corollary}[Asymptotic Normality]\label{co:normality3}
Under Assumptions~\ref{as:setting}, \ref{as:mspace}--\ref{as:weights}, \ref{as:mns}\(_{(d=1/4)}\), \ref{as:bounded-norms}\(_{(r=3)}\), \ref{as:bounded-norms}\(_{(r=4)}\), \ref{as:lowerbound-sigma}, \ref{as:ffm}$(b)$, and \ref{as:conserv},
the aggregate Riesz estimator defined in \eqref{eq:finite-estimator} satisfies
$\tilde{\sigma}_{n}^{-1} \left( \hat{\tau}_n(z, \omega) - \tau_n \right) \xrightarrow{d} \mathcal{N}(0, 1)$.
\end{corollary}
Note that consistency is preserved even if some independent pairs are mistakenly included in $\tilde{\mathcal{E}}_n$, provided that the total number of dependent pairs, including both genuinely dependent and erroneously included independent pairs, remains within the sparsity condition $D_n = o(n^{1/4})$.

In the special case where all units are believed to be mutually independent,
\emph{i.e.}, each \( \zeta_i \) is independent of every other \( \zeta_j \),
\( \tilde{\mathcal{E}}_n \) consists only of diagonal pairs \( (i,i) \).
Then \( \hat{\Sigma}_n^d \) reduces to a diagonal matrix, and the variance estimator and the true variances simplify to
$\tilde{\sigma}_n^2 = \sum_{i=1}^n \nu_{ni}^2 \zeta_i^2$, and
$\sigma_n^2 = \sum_{i=1}^n \nu_{ni}^2 \mathbb{E}[\zeta_i^2]$,
respectively.
This recovers the classical variance estimator and asymptotic normality result under independence as a special case of the locally dependent framework considered here.

From a functional perspective, the target variance $\sigma_n^2$ can be viewed as a continuous linear functional of second-order moments, such as $\Sigma_n$ or the covariances $\operatorname{Cov}(\zeta_i,\zeta_j)$.
Because these moments are not directly observable from a single experimental realisation, the estimator replaces them with observed products $\zeta_i\zeta_j$, yielding the plug-in form in \eqref{eq:plugin-vestimator}.
Consistency follows from structural assumptions on the dependence neighbourhood and boundedness of higher-order terms.

Taken together, the preceding results establish consistent estimation, asymptotic normality, and feasible variance estimation for the aggregate Riesz estimator under local dependence.

\section{Simulation Study}\label{sec:simulation}

This section evaluates the finite-sample behaviour of the aggregate Riesz estimator and the proposed variance estimators under local dependence.
The simulations examine empirical coverage, rejection frequencies, and power under increasing dependence neighbourhood sizes.
Two data-generating mechanisms are considered, a baseline model without spillovers and a stochastic interference model with random exposure effects.
Inference is conducted using the aggregate Riesz estimator together with the variance estimators developed in Section~\ref{sec:variance}.

\comm{Data-generating process.}

We consider two data-generating mechanisms.
Throughout, treatment assignment is generated independently as $z_i\sim\mathrm{Bernoulli}(0.5)$ for each unit $i\in\{1,\dots,n\}$.
Unit-level covariates and parameters are drawn independently according to
\begin{equation}
x_i \stackrel{\text{i.i.d.}}{\sim} \mathcal{N}(0,1),
\qquad
\alpha_i \stackrel{\text{i.i.d.}}{\sim} \mathcal{N}(0,1),
\qquad
\beta_i \stackrel{\text{i.i.d.}}{\sim} \mathcal{N}(1,1).
\end{equation}
Inference is conducted under a sharp-null formulation that specifies unit-level contrasts.
In the size experiments, the centring term is taken to be the true unit-level contrast $\theta_i(\tilde y_i)$.
In the power experiments, we instead impose the null specification $\theta_i(\tilde y_i)\equiv 0$ and evaluate rejection frequencies under this centring.

Local dependence is introduced by partitioning the sample into blocks, where units in the same block share a common latent shock. Specifically, the number of blocks is $B_n = \lfloor n^{1 - d} \rfloor$ for a fixed parameter $d \in \{0, 0.1, 0.2, 0.25, 0.3\}$. The error for unit $i$ is
\begin{equation}
\varepsilon_i = \gamma_i \, \eta_{b(i)} + \nu_i,
\end{equation}
where $b(i)$ denotes the block index to which unit $i$ belongs, $\eta_{b(i)} \sim \mathcal{N}(0, 1)$ is the shared block-level shock, $\nu_i \sim \mathcal{N}(0,1)$ is an idiosyncratic noise term, and $\gamma_i \in \{-1,1\}$ is a random sign (independent Rademacher) to allow both positive and negative correlations between units in the same block.
This induces within-block dependence and across-block independence.

\comm{first case}

In the first case, potential outcomes are generated via
\begin{equation}
\tilde y_i(z,\omega)
=
\alpha_i + \beta_i z_i + \delta_i x_i(\omega) + \varepsilon_i(\omega),
\end{equation}
where $\delta_i$ are drawn \emph{i.i.d.}\ from $\mathcal{N}(0,1)$,
$\varepsilon_i$ is an outcome-level noise term exhibiting local dependence described below.
In this specification, the individual treatment effect satisfies
$\theta_i(\tilde y_i) = \beta_i$,
so that the finite-population FPO and RPO estimands coincide.

\comm{second case}

As an additional illustration, we also consider a stochastic interference design with random exposure effects.
\begin{example}[Network Intervention with Stochastic Spillovers]\label{ex:network-spillover}
Let $G_n(\omega)$ be a random graph on vertex set $\{1,\dots,n\}$, where $\omega$ collects the random edges and possibly additional latent shocks.
Write $A(\omega)$ for its adjacency matrix and assume that $A_{ii}(\omega)=1$ for all $i$.
Let $N_i(\omega)=\{j:A_{ij}(\omega)=1\}$ denote the realised neighbourhood of unit $i$, so that $i\in N_i(\omega)$ and $|N_i(\omega)|\ge 1$.

Consider a treatment assignment vector $z\in\{0,1\}^n$ drawn from a known randomisation design, and define the realised exposure
\begin{equation}\label{eq:exposure}
e_i(z,\omega)=\frac{1}{|N_i(\omega)|}\sum_{j\in N_i(\omega)} z_j,
\end{equation}
that is, the fraction of treated units in the closed neighbourhood of unit $i$, including unit $i$ itself.
A simple stochastic spillover model is
\begin{equation}\label{eq:network-dgp}
\tilde y_i(z,\omega)=\alpha_i+\beta_i z_i+\gamma_i \, e_i(z,\omega)+\varepsilon_i(\omega),
\end{equation}
where $\varepsilon_i(\omega)$ represent idiosyncratic outcome-level randomness, and the spillover term depends on the realised graph through $e_i(z,\omega)$.
\end{example}

Local dependence is introduced through the same block structure as in the baseline design.
A random graph $G_n(\omega)$ is generated independently within each block.
For any pair of distinct units $i$ and $j$ belonging to the same block, an undirected edge between $i$ and $j$ is formed independently with probability $p_{\mathrm{edge}}=0.3$.
No edges are formed between units belonging to different blocks.
$A_{ij}(\omega)=1$ if an edge is present between units $i$ and $j$, and $A_{ij}(\omega)=0$ otherwise.
By construction, $A_{ii}(\omega)=0$ for all $i$.

Given a treatment assignment vector $z\in\{0,1\}^n$, we obtain the realised neighbourhood $N_i(\omega)$ and then compute the realised exposure $e_i(z,\omega)$ according to \eqref{eq:exposure}.
Potential outcomes are generated according to \eqref{eq:network-dgp},
with $\gamma_i = 0.5$ controlling the strength of spillover effects.
Unlike the classical FPO framework, the interference structure itself is stochastic through the latent environment $\omega$, since both the realised neighbourhoods and exposure mappings depend on the random graph $G_n(\omega)$.

In this network setting, the FPO estimand is
\begin{equation}
\tau_n^{\mathrm{FPO}}
= \frac{1}{n} \sum_{i=1}^n \left[
\beta_i + \gamma_i \bigl\{ e_i(z^1,\omega) - e_i(z^0,\omega) \bigr\} \right]
= \frac{1}{n} \sum_{i=1}^n \left[
\beta_i + \frac{\gamma_i}{|N_i(\omega)|} \right]
,
\end{equation}
and the corresponding RPO estimand is
\begin{equation}
\tau_n^{\mathrm{RPO}}
= \frac{1}{n} \sum_{i=1}^n \left[
\beta_i
+
\gamma_i\,\mathbb{E}_\omega\!\bigl\{  e_i(z^1,\omega) - e_i(z^0,\omega) \bigr\} \right]
= \frac{1}{n} \sum_{i=1}^n \left[
\beta_i + \gamma_i \frac{1-(1-p_{\mathrm{edge}})^{m_i}}{m_i\, p_{\mathrm{edge}}} \right],
\end{equation}
where the final equality follows from Proposition~\ref{pp:inv-degree} in Appendix~\ref{ap:proof}
and $m_i$ is the block size.

\comm{Estimation.}

Proposition~\ref{pp:example12} establishes that, under the two-treatment Bernoulli design, the relevant treatment-effect functional is linear and continuous on the outcome model space.
The proof also gives the corresponding Riesz representer, which takes the Horvitz--Thompson form
\begin{equation}
\psi_i(z,\omega_0)
=
\frac{z}{\mu_1}
-
\frac{1-z}{\mu_0},
\end{equation}
where $\mu_1=\mathbb{E}[z_i]=0.5$ and $\mu_0=1-\mu_1$.
The simulations therefore do not involve auxiliary estimation of $\psi_i$.

The estimator is
$\hat{\tau}_n = \frac{1}{n} \sum_{i=1}^n Y_i \psi_i(z_i, \omega_0)$,
and the variance estimator follows \eqref{eq:plugin-vestimator}
with $\zeta_i = Y_i \psi_i(z_i) - \theta_i(\tilde{y}_i)$ and $\nu_{ni} = 1/n$.
That is, we sum pairwise products of residuals over all unit pairs within the same block.

In this simulation setting, the variance estimators \( \hat{\sigma}_{cn}^2 \) and \( \tilde{\sigma}_{n}^2 \) coincide with the local-dependence estimator \( \hat{\sigma}_n^2 \).
This is because the dependency structure is fully captured by shared block-level shocks, and all non-zero cross-moments correspond to block-sharing units. Therefore, all simulation results reflect the performance of both estimators. In general, however, the correlation-based estimator yields improved efficiency in settings where weakly dependent pairs can be excluded without loss of accuracy.

\comm{Simulation design.}

We consider sample sizes $n \in \{100, 200, 500, 1000\}$ and dependency growth parameters $d \in \{0.0, 0.1, 0.2, 0.25, 0.3\}$,
where the largest value $d = 0.3$ lies outside the range permitted by Assumption~\ref{as:mns} $(d \le 1/4)$ in Theorem~\ref{th:normality}. 
For each $(n,d)$ pair, we run $2000$ independent replications. 
For each replication, we compute the estimator $\hat{\tau}_n$ and the standard error $\hat{\sigma}_n$, and conduct inference under two specifications of the centring term, the correctly specified sharp null $\theta_i(\tilde y_i)$ and the misspecified null $\theta_i(\tilde y_i)\equiv 0$. 
We record the following metrics:
\begin{itemize}
\item Empirical coverage of 95\% confidence intervals, defined as the fraction of replications satisfying $|\hat{\tau}_n - \tau_n| \leq 1.96 \hat{\sigma}_n$, under the correct specification;
\item Rejection rates at the 1\%, 5\%, and 10\% levels for both specifications.
\end{itemize}

\comm{Coverage.}

Figures~\ref{fig:coverage1} and~\ref{fig:coverage2} report the empirical coverage of nominal 95\% confidence intervals across different dependency growth rates $d$ and sample sizes $n$.

\begin{figure}[h!]
\centering
\input{plots/coverage_multiN.tex}
\caption{Empirical coverage of 95\% confidence intervals under local dependence (baseline design).}
\label{fig:coverage1}
\end{figure}

\begin{figure}[h!]
\centering
\input{plots/coverage_multiN3size.tex}
\caption{Empirical coverage of 95\% confidence intervals under network interference and stochastic spillovers.}
\label{fig:coverage2}
\end{figure}

Figure~\ref{fig:coverage1} corresponds to the baseline design in which the FPO and RPO estimands coincide.
Across all configurations, empirical coverage is slightly above the nominal $0.95$, indicating mild finite-sample conservativeness.
For small samples ($n\in \{100,200\}$), coverage shows no clear monotone pattern in $d$, consistent with higher sampling variability.
For larger $n$, especially $n=1000$, coverage increases gradually with $d$, reflecting the growing conservativeness of the variance estimator as dependence neighbourhoods expand.
Within the theoretical regime $d\le 1/4$ (Theorem~\ref{th:normality}), coverage remains close to nominal and uniformly satisfactory.

Figure~\ref{fig:coverage2} reports coverage under network interference.
Coverage is uniformly above 95\%, indicating stronger conservativeness than in the baseline case.
This persists even at $d=0$, suggesting that interference through the exposure mapping $e_i(z,\omega)$ affects calibration beyond what is captured by neighbourhood growth alone.
When dependence exceeds the theoretical range ($d=0.3$), coverage approaches 97\%, consistent with a gradual shift towards conservativeness as assumptions are violated.
Overall, coverage remains stable and well controlled, supporting the feasibility of design-based inference for $\tau_n$ under network interference.

\comm{Test size and power}

Tables~\ref{tab:rejection1} and~\ref{tab:rejection2} report empirical rejection frequencies at the 1\%, 5\%, and 10\% nominal levels for the size and power experiments, across all combinations of sample size $n$ and dependency growth rate $d$, for the two data-generating mechanisms described above.

\begin{table}[!ht]
\centering
\caption{ Empirical rejection rates under local dependence (baseline design).
Results are based on 2000 Monte Carlo replications for each $(n,d)$ configuration.}
\label{tab:rejection1}

\begin{tabular}{cc cccc cccc}
\toprule
& & \multicolumn{4}{c}{\textbf{Size experiment}} 
  & \multicolumn{4}{c}{\textbf{Power experiment}} \\
\cmidrule(lr){3-6}\cmidrule(lr){7-10}
\textbf{level} & \textbf{$d$} 
& $n=100$ & $200$ & $500$ & $1000$
& $n=100$ & $200$ & $500$ & $1000$ \\
\midrule

\multirow{5}{*}{1\%}
& 0.00 & 0.0095 & 0.0090 & 0.0070 & 0.0085 & 0.340 & 0.730 & 0.997 & 1.000 \\
& 0.10 & 0.0090 & 0.0065 & 0.0085 & 0.0110 & 0.335 & 0.741 & 0.994 & 1.000 \\
& 0.20 & 0.0050 & 0.0065 & 0.0070 & 0.0075 & 0.314 & 0.705 & 0.991 & 1.000 \\
& 0.25 & 0.0050 & 0.0095 & 0.0060 & 0.0050 & 0.285 & 0.710 & 0.990 & 1.000 \\
& 0.30 & 0.0060 & 0.0065 & 0.0055 & 0.0055 & 0.268 & 0.690 & 0.992 & 1.000 \\

\midrule
\multirow{5}{*}{5\%}
& 0.00 & 0.0445 & 0.0395 & 0.0415 & 0.0480 & 0.619 & 0.894 & 1.000 & 1.000 \\
& 0.10 & 0.0425 & 0.0360 & 0.0410 & 0.0435 & 0.618 & 0.902 & 1.000 & 1.000 \\
& 0.20 & 0.0400 & 0.0465 & 0.0460 & 0.0435 & 0.600 & 0.892 & 0.999 & 1.000 \\
& 0.25 & 0.0415 & 0.0420 & 0.0450 & 0.0420 & 0.597 & 0.894 & 0.998 & 1.000 \\
& 0.30 & 0.0495 & 0.0415 & 0.0400 & 0.0360 & 0.572 & 0.885 & 0.998 & 1.000 \\

\midrule
\multirow{5}{*}{10\%}
& 0.00 & 0.0865 & 0.0905 & 0.0890 & 0.0935 & 0.744 & 0.942 & 1.000 & 1.000 \\
& 0.10 & 0.0975 & 0.0875 & 0.0845 & 0.0875 & 0.740 & 0.948 & 1.000 & 1.000 \\
& 0.20 & 0.0860 & 0.1040 & 0.0890 & 0.0850 & 0.714 & 0.938 & 1.000 & 1.000 \\
& 0.25 & 0.0950 & 0.0960 & 0.0950 & 0.0905 & 0.732 & 0.942 & 1.000 & 1.000 \\
& 0.30 & 0.0980 & 0.0905 & 0.0940 & 0.0870 & 0.719 & 0.944 & 1.000 & 1.000 \\

\bottomrule
\end{tabular}
\end{table}

\begin{table}[htbp]
\centering
\caption{Empirical rejection rates under network interference with stochastic spillovers.
Results are based on 2000 Monte Carlo replications for each $(n,d)$ configuration.}
\label{tab:rejection2}
\begin{tabular}{cc cccc cccc}
\toprule
& & \multicolumn{4}{c}{\textbf{Size experiment}} 
  & \multicolumn{4}{c}{\textbf{Power experiment}} \\
\cmidrule(lr){3-6}\cmidrule(lr){7-10}
\textbf{level} & \textbf{$d$} 
& $n=100$ & $200$ & $500$ & $1000$
& $n=100$ & $200$ & $500$ & $1000$ \\
\midrule

\multirow{5}{*}{1\%}
& 0.00 & 0.0070 & 0.0060 & 0.0030 & 0.0090 & 0.872 & 0.997 & 1.000 & 1.000 \\
& 0.10 & 0.0055 & 0.0055 & 0.0070 & 0.0055 & 0.781 & 0.992 & 1.000 & 1.000 \\
& 0.20 & 0.0040 & 0.0035 & 0.0050 & 0.0030 & 0.707 & 0.979 & 1.000 & 1.000 \\
& 0.25 & 0.0020 & 0.0045 & 0.0045 & 0.0060 & 0.621 & 0.977 & 1.000 & 1.000 \\
& 0.30 & 0.0025 & 0.0040 & 0.0025 & 0.0035 & 0.638 & 0.954 & 1.000 & 1.000 \\

\midrule
\multirow{5}{*}{5\%}
& 0.00 & 0.0375 & 0.0365 & 0.0345 & 0.0405 & 0.968 & 1.000 & 1.000 & 1.000 \\
& 0.10 & 0.0360 & 0.0305 & 0.0325 & 0.0285 & 0.938 & 0.999 & 1.000 & 1.000 \\
& 0.20 & 0.0335 & 0.0360 & 0.0350 & 0.0315 & 0.917 & 0.998 & 1.000 & 1.000 \\
& 0.25 & 0.0340 & 0.0270 & 0.0350 & 0.0340 & 0.875 & 0.998 & 1.000 & 1.000 \\
& 0.30 & 0.0305 & 0.0310 & 0.0330 & 0.0305 & 0.886 & 0.996 & 1.000 & 1.000 \\

\midrule
\multirow{5}{*}{10\%}
& 0.00 & 0.0760 & 0.0770 & 0.0730 & 0.0770 & 0.988 & 1.000 & 1.000 & 1.000 \\
& 0.10 & 0.0725 & 0.0805 & 0.0800 & 0.0710 & 0.974 & 1.000 & 1.000 & 1.000 \\
& 0.20 & 0.0725 & 0.0820 & 0.0780 & 0.0715 & 0.960 & 1.000 & 1.000 & 1.000 \\
& 0.25 & 0.0850 & 0.0680 & 0.0720 & 0.0720 & 0.941 & 1.000 & 1.000 & 1.000 \\
& 0.30 & 0.0735 & 0.0790 & 0.0720 & 0.0780 & 0.949 & 1.000 & 1.000 & 1.000 \\

\bottomrule
\end{tabular}
\end{table}

In the baseline setting, empirical rejection frequencies in the size experiment remain close to their nominal levels across all sample sizes and dependency regimes.
At the 5\% level, rejection rates typically lie between approximately 4\% and 5\%, with similar behaviour observed at the 1\% and 10\% levels, and no systematic drift as $n$ increases.
These results indicate that the proposed procedure maintains appropriate size control under local dependence when the conditions of Corollary~\ref{co:normality} are satisfied.

Across most configurations, rejection frequencies are slightly below nominal levels.
This behaviour is consistent with the coverage results in Figure~\ref{fig:coverage1} and suggests that conservativeness is driven primarily by variance estimation rather than estimator bias.
When the dependency growth rate exceeds the theoretical threshold in Assumption~\ref{as:mns}, mild over-rejection emerges, most notably at $d=0.3$, consistent with the requirement that dependence neighbourhoods grow sufficiently slowly.

In the power experiment, rejection frequencies increase monotonically with sample size.
For moderate $n$, power remains substantial even under stronger dependence, and for $n\ge 500$ rejection rates are essentially one across all significance levels, confirming that the procedure achieves high power without compromising size control within the theoretical regime.

In the network-interference setting, size performance remains stable and close to nominal levels across configurations.
As in the baseline design, slight under-rejection is observed, mirroring the coverage results in Figure~\ref{fig:coverage2}.
Notably, stochastic spillovers and outcome-level randomness do not introduce additional size distortions relative to the baseline case.

Power properties remain strong under network interference.
Rejection frequencies increase rapidly with $n$, and for $n\ge 200$ power is close to one at conventional significance levels.
As in the baseline design, power decreases modestly with stronger dependence for small $n$, reflecting information loss due to local dependence.
Overall, the results demonstrate that the proposed procedure retains high power while maintaining accurate size control when targeting $\tau_n$ under network interference.

\comm{overall}

Overall, the simulation results are consistent with the asymptotic theory developed in Sections~\ref{sec:asymptotic} and~\ref{sec:variance}.
Within the dependence regimes covered by the theoretical assumptions, the aggregate Riesz estimator maintains stable empirical coverage and accurate size control.
Power increases with sample size across all dependence configurations.
The results also indicate that the proposed variance estimators remain stable under moderate local dependence.

\section{Conclusion}\label{sec:conclusion}

Absent suitable probabilistic structure, a single realised experiment does not contain sufficient information to recover expectation-based causal estimands defined over latent stochastic environments.

This paper studies conditions under which such recoverability becomes possible within a design-based framework under RPO.
Potential outcomes are modelled as random functions of a latent stochastic environment, and causal estimands are defined as expectations over the induced stochastic mechanism.
Treatment-effect functionals admitting continuous linear representations are identified through the randomisation design via their associated Riesz representers.
Under suitable local dependence and neighbourhood growth conditions, cross-sectional averaging asymptotically substitutes for averaging over the latent stochastic environment.
This yields consistency and asymptotic normality of the aggregate Riesz estimator for $\tau_n$ from a single realised experiment.

The paper further establishes feasible variance estimation under local dependence.
The proposed variance estimators are consistent for the true sampling variance and therefore permit asymptotically valid inference without relying on conservative finite-population variance bounds.

The simulation results support the asymptotic theory.
Within the theoretical range of neighbourhood growth, the proposed procedure exhibits stable empirical coverage and accurate size control under local dependence.

The framework developed in this paper is intended for randomised experiments under structured stochastic dependence.
More broadly, the paper reframes design-based inference under RPO as a recoverability problem for mechanism-level causal quantities, where suitable probabilistic structure induced by the experimental design permits asymptotically valid inference from a single realised experiment.

\comm{citation not in text}
\nocite{riesz1907}

\comm{citations in Appendxi}
\nocite{chen2007sieve, chen_pouzo2012moment, van2000asymptotic}
\nocite{adams2003sobolev}
\nocite{billingsley1999convergence}
\nocite{Chernozhukov02122025}

\subsection*{Acknowledgements}

I am grateful to Fredrik Sävje for his valuable discussions and support.

\subsection*{Data Availability Statement}

No external or empirical datasets were used in this study.
All data arise from simulations conducted by the author.
The complete codebase, including replication scripts and plotting routines, is publicly available at
\href{https://github.com/yukai-yang/RieszRE_Experiments}{\url{github.com/yukai-yang/RieszRE_Experiments}} under the MIT license.

\bibliographystyle{chicago}
\bibliography{references}

\newpage

\appendix

\section{Notation Summary}

\renewcommand{\arraystretch}{1.3}
\begin{table}[h!]
\centering
\caption{Summary of notation used in the paper.}
\begin{tabular}{ll}
\toprule
\textbf{Symbol} & \textbf{Description} \\
\midrule
\( n \) & Number of units in the experiment \\
\( z \in \mathcal{Z} \) & Treatment assignment vector \\
\( \mu \) & Probability measure for randomisation \\
\( \omega \in \Omega \) & Latent variable representing outcome-level randomness \\
\( P \) & Probability measure for the latent random variable $\omega$ \\
\( \mathcal{F}_x \) & Sigma-algebra generated by the random variable $x$ \\
\( \tilde{y}_i(z, \omega) \) & Potential outcome of unit \( i \) under treatment \( z \) and latent variable \( \omega \) \\
\( \theta_i(\tilde{y}_i) \) & Linear functional defining the treatment effect for unit \( i \) \\
\( \hat{\theta}_i(z, \omega) \) & Riesz estimator for the treatment effect for unit \( i \) \\
\( \zeta_i \) & The difference \( \hat{\theta}_i(z,\omega) - \theta_i(\tilde{y}_i) \) \\
\( \tau_n \) & Aggregate estimand (e.g., average treatment effect) \\
\( \hat{\tau}_n \) & Aggregate Riesz estimator for \( \tau_n \) \\
\( \psi_i(z, \omega) \) & Riesz representer for unit \( i \) \\
\( \hat{\psi}_i(z, \omega) \) & Estimated Riesz representer (e.g., via basis expansion) \\
\( \mathcal{M}_i \) & Model space for unit \( i \)'s outcome function \\
\( N_i \) & Dependency neighbourhood of unit \( i \) \\
\( D_n \) & Maximum neighbourhood size across units \\
\( d_n \) & Average neighbourhood size across units \\
\( \sigma_n^2 \) & Variance of \( \hat{\tau}_n \) \\
\( \mathcal{E}_n \) & Set of dependent pairs $(i,j)$ for $\zeta_i$ and $\zeta_j$ \\
\( \hat{\sigma}_n^2 \) & Variance estimator of \( \hat{\tau}_n \) using \( \mathcal{E}_n \) \\
\( \mathcal{E}_n^c \) & Set of correlated pairs $(i,j)$ for $\zeta_i$ and $\zeta_j$ \\
\( \hat{\sigma}_{cn}^2 \) & Variance estimator of \( \hat{\tau}_n \) using \( \mathcal{E}_n^c \) \\
\( \tilde{\mathcal{E}}_n \) & Conservative set of index pairs \\
\( \tilde{\sigma}_n^2 \) & Variance estimator of \( \hat{\tau}_n \) using \( \tilde{\mathcal{E}}_n \) \\
\bottomrule
\end{tabular}
\label{tab:notation}
\end{table}

\section{Lemmas and Proofs} \label{ap:proof}

\begin{proof}[Proof of Proposition~\ref{pp:example12}]
The linearity for both the two examples follows directly from the facts that integrals are linear and that we are taking differences of linear functionals.

For continuity, we need to rewrite \eqref{eq:example2} clearly
\begin{equation}
\theta_i(u)
= \iint_{\mathcal{Z}\times\Omega}\left(\frac{\mathbf{1}_A(z)}{\mu(A)} - \frac{\mathbf{1}_B(z)}{\mu(B)}\right)u(z,\omega)\,\mu(\mathrm{d}z)\,P(\mathrm{d}\omega) = \innp{u}{\psi},
\end{equation}
and
\begin{equation}
\psi(z, \omega) = \frac{\mathbf{1}_A(z)}{\mu(A)} - \frac{\mathbf{1}_B(z)}{\mu(B)}.
\end{equation}
We can see that $\psi\in L^2(\mathcal{Z}\times\Omega)$ as
\begin{eqnarray*}
\|\psi\|_{L^2}^2
&=& \int_\Omega\int_\mathcal{Z}\left|\frac{\mathbf{1}_A(z)}{\mu(A)} - \frac{\mathbf{1}_B(z)}{\mu(B)}\right|^2\,\mu(\dif z)\,P(\dif \omega)\\
&=& \int_{\mathcal{Z}}\left|\frac{\mathbf{1}_A(z)}{\mu(A)} - \frac{\mathbf{1}_B(z)}{\mu(B)}\right|^2\,\mu(\mathrm{d}z)
= \frac{1}{\mu(A)}+\frac{1}{\mu(B)} < \infty.
\end{eqnarray*}
Now continuity is immediate due to the Cauchy--Schwarz inequality 
\begin{equation*}
|\theta_i(u)| = |\langle u,\psi\rangle| \leq \norm{u} \, \norm{\psi}.
\end{equation*}
\end{proof}

\begin{proof}[Proof of Proposition~\ref{pp:point_derivative}]
Linearity follows from the linearity of differentiation and expectation.

Since $s>1+\frac d2$, the Sobolev embedding theorem yields a continuous embedding
$H^s(\mathcal Z)\hookrightarrow C^1(\overline{\mathcal Z})$; see, e.g.,
\citet[Chapter~5]{adams2003sobolev}. Hence, for almost every $\omega\in\Omega$,
the function $z\mapsto u(z,\omega)$ admits a $C^1(\overline{\mathcal Z})$ representative,
so that $\partial_z u(z_0,\omega)$ is well-defined. Moreover, by continuity of the embedding,
there exists a constant $C>0$ such that
\[
\bigl|\partial_z u(z_0,\omega)\bigr|\le C\|u(\cdot,\omega)\|_{H^s(\mathcal Z)}
\quad\text{for a.e. }\omega.
\]
Therefore,
\begin{align*}
|\theta_i(u)|
&= \left|\mathbb E_\omega\!\left[\partial_z u(z_0,\omega)\right]\right|
\le \mathbb E_\omega\!\left[\bigl|\partial_z u(z_0,\omega)\bigr|\right] \\
&\le C\,\mathbb E_\omega\!\left[\|u(\cdot,\omega)\|_{H^s(\mathcal Z)}\right]
\le C\Bigl(\mathbb E_\omega\|u(\cdot,\omega)\|_{H^s(\mathcal Z)}^2\Bigr)^{1/2}
= C\|u\|_{L^2(\Omega;H^s(\mathcal Z))},
\end{align*}
where the last inequality is Cauchy--Schwarz. This proves continuity.
\end{proof}

\begin{proof}[Proof of Proposition~\ref{pp:region_integral}]
Linearity is immediate.
For a.e. $\omega$, $\partial_z u(\cdot,\omega)\in L^2(\mathcal Z)$. Hence the integral is well-defined.
By Cauchy--Schwarz,
\begin{equation*}
\left|\int_{\mathcal Z}\partial_z u(z,\omega)\,\varphi(z)\,dz\right|
\le \|\partial_z u(\cdot,\omega)\|_{L^2(\mathcal Z)}\|\varphi\|_{L^2(\mathcal Z)}.
\end{equation*}
Taking expectations and applying Cauchy--Schwarz in $\omega$ yields
\begin{equation*}
|\theta_i(u)|
\le \|\varphi\|_{L^2(\mathcal Z)}\Bigl(\mathbb E\|\partial_z u(\cdot,\omega)\|_{L^2(\mathcal Z)}^2\Bigr)^{1/2}
\le \|\varphi\|_{L^2(\mathcal Z)}\,\|u\|_{L^2(\Omega;H^1(\mathcal Z))}.
\end{equation*}
\end{proof}

\begin{proof}[Proof of Theorem~\ref{th:Riesz}]
The result follows directly from the Riesz representation theorem.
\end{proof}

\begin{proof}[Proof of Theorem~\ref{th:unbiasedness}]
By definition of the Riesz estimator in \eqref{eq:Riesz-estimator-i} and the linearity of expectation, we have
\[
\mathbb{E}\bigl[\hat{\theta}_i(z, \omega)\bigr] = \theta_i(\tilde{y}_i),
\]
which further implies
\[
\mathbb{E}\bigl[\hat{\tau}_n(z, \omega)\bigr]
= \sum_{i=1}^n \nu_{ni} \, \mathbb{E}\bigl[\hat{\theta}_i(z, \omega)\bigr]
= \sum_{i=1}^n \nu_{ni} \, \theta_i(\tilde{y}_i)
= \tau_n.
\]
This equality follows from the construction of the estimator and does not require any integrability assumptions. In particular, it holds even if $\tilde{y}_i(z, \omega)\, \psi_i(z, \omega) \notin L^1$.
\end{proof}

\begin{lemma}\label{lm:nu-bound}
Suppose the weights \( \nu_{ni} \) satisfy Assumption~\ref{as:weights}. Then for any integer \( r \geq 1 \), the following bound holds:
\begin{equation}\label{eq:sumnu-O}
\sum_{i=1}^{n} \nu_{ni}^r = O(n^{1 - r}).
\end{equation}
\end{lemma}

\begin{proof}
By Assumption~\ref{as:weights}, we have
\[
\sum_{i=1}^{n} \nu_{ni}^r \le \sum_{i=1}^{n} \left( \frac{\bar{\nu}}{n} \right)^r = n \cdot \left( \frac{\bar{\nu}}{n} \right)^r = \bar{\nu}^r \cdot n^{1 - r},
\]
which implies the claimed order bound.
\end{proof}

\begin{lemma}\label{lm:r-bound}
Under Assumptions~\ref{as:setting}, \ref{as:mspace}, and \ref{as:dual}, 
the \( r \)th central moment of the treatment effect Riesz estimator,
for some integer $r \geq 1$, if it exists, satisfies
\begin{equation}\label{eq:r-bound1}
\mathbb{E} \left| \hat{\theta}_i(z, \omega) - \theta_i(\tilde{y}_i) \right|^r 
\leq 
\left( 2  \norm{\tilde{y}_i}_p \, \norm{\psi_i}_q \right)^r,
\end{equation}
for any \( p, q \in [1, \infty] \) satisfying \( 1/p + 1/q = 1/r \).
In particular, when $r=2$,
\begin{equation}\label{eq:r-bound2}
\mathbb{E} \left[ \hat{\theta}_i(z, \omega) - \theta_i(\tilde{y}_i) \right]^2 
\leq 
\norm{\tilde{y}_i}_p^2 \, \norm{\psi_i}_q^2.
\end{equation}
The inequalities \eqref{eq:r-bound1} and \eqref{eq:r-bound2} hold trivially if either $\| \tilde{y}_i \|_p$ or $\| \psi_i \|_q$ has no finite upper bound.
\end{lemma}

\begin{proof}
The argument parallels Lemma C.3 in \citet{harshaw2022riesz}. 
Applying Hölder's inequality on the product space \( \mathcal{Z} \times \Omega \) yields
\[
\mathbb{E} \left| \tilde y_i(z,\omega)\psi_i(z,\omega) \right|^r
\le
\|\tilde y_i\|_p^r \|\psi_i\|_q^r,
\]
for any \( p,q \in [1,\infty] \) satisfying \(1/p+1/q=1/r\).
The stated bound then follows from the triangle inequality together with the representation
\[
\theta_i(\tilde y_i)
=
\mathbb E\!\left[\tilde y_i(z,\omega)\psi_i(z,\omega)\right].
\]


Regarding inequality \eqref{eq:r-bound2},
\begin{align*}
    \mathbb{E} \left[ \hat{\theta}_i(z, \omega) - \theta_i(\tilde{y}_i) \right]^2 
&= \mathbb{E}[\hat{\theta}_i(z, \omega)]^2 - (\theta_i(\tilde{y}_i))^2 
\leq \mathbb{E}\left[ \tilde{y}_i(z, \omega) \psi_i(z, \omega) \right]^2 \\
&= \norm{\tilde{y}_i(z, \omega) \psi_i(z, \omega)}^2
\leq \left( \norm{\tilde{y}_i(z, \omega)}_p \, \norm{\psi_i(z, \omega)}_q \right)^2.
\end{align*}
\end{proof}

We now state a proposition that provides an upper bound on the variance
of the aggregate Riesz estimator $\hat{\tau}_n(z, \omega)$.
This result follows from the setup and notation introduced above, and it serves to justify or at least suggest the consistency of the estimator under appropriate conditions.
\begin{proposition}[Variance Upper Bound]\label{pp:sigma2}
Under Assumptions~\ref{as:setting} and \ref{as:mspace}--\ref{as:weights}, 
the following inequality holds
\begin{equation}\label{eq:sigma2-bound}
\sigma_n^2 \leq \frac{\bar{\nu}^2}{n^2} \sum_{i=1}^{n} \abs{N_i} \, \norm{\tilde{y}_i}_p^2 \, \norm{\psi_i}_q^2,
\end{equation}
for any \( p, q \in [1, \infty] \) satisfying \( 1/p + 1/q = 1/2 \).
The inequality holds trivially if either $\| \tilde{y}_i \|_p$ or $\| \psi_i \|_q$ has no finite upper bound.

Furthermore, we define $S_2 := \left\{ (p, q) \in [1, \infty]^2 :\, \frac{1}{p} + \frac{1}{q} = \frac{1}{2} \right\}$,
and then we have
\begin{equation}\label{eq:sigma2-unibound}
\sigma_n^2 \leq \frac{\bar{\nu}^2 d_n}{n} \left( \inf_{(p, q) \in S_2} \, \norm{\tilde{y}_i}_{\max, p}^n \, \norm{\psi_i}_{\max, q}^n\right)^2.
\end{equation}
The inequality hold trivially if $\inf_{p, q \in S} \, \norm{\tilde{y}_i}_{\max, p}^n \, \norm{\psi_i}_{\max, q}^n$ has no finite upper bound.
\end{proposition}

\begin{proof}[Proof of Proposition~\ref{pp:sigma2}]
Let $\zeta_i = \hat{\theta}_i(z, \omega) - \theta_i(\tilde{y}_i)$ and we observe that $\mathbb{E}[\zeta_i] = 0$ by Theorem~\ref{th:unbiasedness}.
Then,
\begin{align*}
\sigma_n^2 &= \mathbb{V}[\hat{\tau}_n(z, \omega) - \tau_n] = \mathbb{V}\left[\sum_{i=1}^{n} \nu_{ni} \zeta_i \right]
= \mathbb{E}\left[ \sum_{i=1}^{n} \nu_{ni} \zeta_i \right]^2 \\
&= \sum_{i=1}^{n} \sum_{j=1}^{n} \nu_{ni} \nu_{nj}
\mathbf{1}(i \in N_j) \mathbf{1}(j \in N_i)
\mathbb{E} \left[ \zeta_i \zeta_j \right] \\
&\le \sum_{i=1}^{n} \sum_{j=1}^{n} \nu_{ni} \nu_{nj}
\mathbf{1}(i \in N_j) \mathbf{1}(j \in N_i)
\norm{\zeta_i} \, \norm{\zeta_j} &&\text{(Cauchy--Schwarz)} \\
&\leq \sum_{i=1}^{n} \nu_{ni} \sum_{j=1}^{n} \frac{1}{2} \nu_{nj}
\left( \mathbf{1}(j \in N_i) \norm{\zeta_i}^2 +
\mathbf{1}(i \in N_j) \norm{\zeta_j}^2 \right) &&\text{(AM-GM)} \\
&= \sum_{i=1}^{n} \nu_{ni} \left(\sum_{j \in N_i} \nu_{nj}\right) \norm{\zeta_i}^2
\leq \frac{\bar{\nu}^2}{n^2} \sum_{i=1}^{n} \abs{N_i} \norm{\zeta_i}^2 &&\text{(Lemma~\ref{lm:nu-bound})}\\
&\leq \frac{\bar{\nu}^2}{n^2} \sum_{i=1}^{n} \abs{N_i} \, \norm{\tilde{y}_i}_p^2 \, \norm{\psi_i}_q^2 &&\text{(Lemma~\ref{lm:r-bound})} 
\end{align*}
Note that AM-GM stands for arithmetic-geometric mean inequality.

To obtain the uniform bound, we apply the max-\( p \) norm, and note that \( d_n = n^{-1} \sum_{i=1}^{n} |N_i| \), yielding
\begin{align*}
\sigma_n^2 \leq \frac{\bar{\nu}^2}{n^2} \left(\norm{\tilde{y}_i}_{\max, p}^n \, \norm{\psi_i}_{\max, q}^n\right)^2 \sum_{i=1}^{n} \abs{N_i} = \frac{\bar{\nu}^2 d_n}{n} \left(\norm{\tilde{y}_i}_{\max, p}^n \, \norm{\psi_i}_{\max, q}^n\right)^2.
\end{align*}
Since this inequality holds for any \( p, q \in [1, \infty] \) satisfying \(1/p + 1/q = 1/2\),
we have a more conservative upper bound
\begin{align*}
\sigma_n^2 \leq \frac{\bar{\nu}^2 d_n}{n} \left( \inf_{p, q \in S_2} \, \norm{\tilde{y}_i}_{\max, p}^n \, \norm{\psi_i}_{\max, q}^n\right)^2.
\end{align*}
\end{proof}

Proposition~\ref{pp:sigma2} provides an upper bound on the variance of the aggregate estimator. It shows that the variance decreases with sample size under boundedness of unit-level norms and local sparsity. The result highlights the interaction between dependence complexity and unit-level variability in determining inferential precision.

It implies that the variance may decay as the sample size increases,
provided that the norms $\norm{\tilde{y}_i}_{\max, p}^n$ and $\norm{\psi_i}_{\max, q}^n$ remain bounded if there exists such a pair $p, q \in [1, \infty]$ satisfying $1/p + 1/q = 1/2$,
and that the neighbourhood sizes $\abs{N_i}$ are uniformly small.
The bound highlights the importance of controlling both the complexity of dependence and the magnitude of unit-level variation in order to obtain valid inference.

\begin{proof}[Proof of Theorem~\ref{th:consistency}]
The result follows directly from Proposition~\ref{pp:sigma2}.
Specifically, the bound in \eqref{eq:sigma2-unibound}, together with Assumption~\ref{as:bounded-norms} with \( r = 2 \), implies that
\[
\mathbb{E}\left[\hat{\tau}_n(z, \omega) - \tau_n\right]^2 \to 0,
\]
as \( n \to \infty \), establishing mean square consistency.

Moreover, suppose $\sup_{n \in \mathbb{N}} d_n < \infty$, that is, there exists some $d_0 \in (0, \infty)$ such that $d_n \leq d_0$ for all $n$.
Since $\hat{\tau}_n - \tau_n$ follows some distribution with zero mean and variance $\sigma_n^2$,
it follows that $\sqrt{n}(\hat{\tau}_n - \tau_n)$ has variance $n\sigma_n^2$.
By the bound in Proposition~\ref{pp:sigma2}, we obtain
\begin{align*}
n\sigma_n^2 \leq \bar{\nu}^2 d_0 \left( \inf_{p, q \in S_2} \, \norm{\tilde{y}_i}_{\max, p}^n \, \norm{\psi_i}_{\max, q}^n\right)^2 < \infty,
\end{align*}
under the stated assumptions.
Hence, the aggregate Riesz estimator satisfies $O_p(n^{-1/2})$.
\end{proof}

\begin{lemma}\label{lm:wasserstein}
Let \( (X, d) \) be a metric space. Suppose \( W_n \) and \( Z \) are random variables taking values in \( X \), and \( d_W(W_n, Z) \to 0 \) as \( n \to \infty \), where \( d_W \) denotes the Wasserstein distance. Then \( W_n \xrightarrow{d} Z \), \emph{i.e.}, \( W_n \) converges in distribution to \( Z \).
\end{lemma}
\begin{proof}
We will show that \( \mathbb{E}[f(W_n)] \to \mathbb{E}[f(Z)] \) for all bounded uniformly continuous functions \( f: X \to \mathbb{R} \), which characterises weak convergence; see \citet[Theorem 2.1]{billingsley1999convergence}.

Let \( f \in \mathrm{BUC}(X) \), the space of all bounded uniformly continuous functions.  
Since Lipschitz functions are dense in \( \mathrm{BUC}(X) \) under the supremum norm, for any \( \varepsilon > 0 \), there exists a Lipschitz function \( f_\varepsilon: X \to \mathbb{R} \) such that
\[
\|f - f_\varepsilon\|_\infty < \frac{\varepsilon}{3}.
\]

The function \( f_\varepsilon \) is Lipschitz with some finite constant \( L_\varepsilon \). Therefore,
\[
\left| \mathbb{E}[f_\varepsilon(W_n)] - \mathbb{E}[f_\varepsilon(Z)] \right| \leq L_\varepsilon \cdot d_W(W_n, Z).
\]
Since \( d_W(W_n, Z) \to 0 \), there exists \( N(\varepsilon) \in \mathbb{N} \) such that
\[
\left| \mathbb{E}[f_\varepsilon(W_n)] - \mathbb{E}[f_\varepsilon(Z)] \right| < \frac{\varepsilon}{3}, \quad \text{for all } n > N(\varepsilon).
\]

Using the triangle inequality, we have
\[
\begin{aligned}
\left| \mathbb{E}[f(W_n)] - \mathbb{E}[f(Z)] \right|
&\leq \left| \mathbb{E}[f(W_n)] - \mathbb{E}[f_\varepsilon(W_n)] \right| \\
&\quad + \left| \mathbb{E}[f_\varepsilon(W_n)] - \mathbb{E}[f_\varepsilon(Z)] \right| \\
&\quad + \left| \mathbb{E}[f_\varepsilon(Z)] - \mathbb{E}[f(Z)] \right| \\
&< \frac{\varepsilon}{3} + \frac{\varepsilon}{3} + \frac{\varepsilon}{3} = \varepsilon.
\end{aligned}
\]

Since \( \varepsilon > 0 \) and \( f \in \mathrm{BUC}(X) \) were arbitrary, we conclude that \( \mathbb{E}[f(W_n)] \to \mathbb{E}[f(Z)] \), \emph{i.e.}, \( W_n \xrightarrow{d} Z \).
\end{proof}

\begin{proof}[Proof of Theorem~\ref{th:normality}]
The proof proceeds in two steps. 
First, we establish Wasserstein convergence using the normal approximation bound for dependency graphs of \citet[Theorem 3.5]{ross2011fundamentals}. 
Second, we show that Wasserstein convergence implies weak convergence.

Clearly $\mathbb{E}[\hat{\theta}_i(z, \omega)]^4 < \infty$ is equivalent to $\mathbb{E}[\hat{\theta}_i(z, \omega) - \theta_i(\tilde{y}_i)]^4 < \infty$.
Define $X_i = \nu_{ni} (\hat{\theta}_i(z, \omega) - \theta_i(\tilde{y}_i))$,
and then $\mathbb{E}[X_i^4] < \infty$.
By Theorem~\ref{th:unbiasedness},
we have also $\mathbb{E}[X_i] = 0$.

Note that $\sum_{i=1}^n X_i = \hat{\tau}_n(z, \omega) - \tau_n$.
Define $W_n = \sum_{i=1}^n X_i / \sigma_n$.
Since the collection $\{X_1, X_2, ..., X_n\}$ have dependency neighborhoods $N_i$, $i = 1, ..., n$, with $D_n = \max_{1 \leq i \leq n} |N_i|$.
Then we obtain
\begin{equation}\label{eq:dw-upper-bound}
d_W(W_n, Z) \leq \frac{D_n^2}{\sigma_n^3} \sum_{i=1}^n \mathbb{E} |X_i|^3 + \frac{\sqrt{26}D_n^{3/2}}{\sqrt{\pi} \sigma_n^2} \sqrt{\sum_{i=1}^n \mathbb{E}[X_i^4]},
\end{equation}
where $d_W(W_n, Z)$ denotes the Wasserstein distance between $W_n$ and a standard normal one $Z$.
Observe that \eqref{eq:dw-upper-bound} gives an upper bound of the distance.

First, we examine the first term in \eqref{eq:dw-upper-bound}

\begin{align*}
\sum_{i=1}^{n} \mathbb{E} \abs{X_i}^3 &=  \sum_{i=1}^{n} \mathbb{E} \abs{\nu_{ni} \left(\hat{\theta}_i(z, \omega) - \theta_i(\tilde{y}_i)\right)}^3 \\
&= \sum_{i=1}^{n} \nu_{ni}^3 \mathbb{E} \abs{\hat{\theta}_i(z, \omega) - \theta_i(\tilde{y}_i)}^3 \\
&\leq \sqrt{\sum_{i=1}^{n} \nu_{ni}^6} \cdot \sqrt{\sum_{i=1}^{n} \left(\mathbb{E} \abs{\hat{\theta}_i(z, \omega) - \theta_i(\tilde{y}_i)}^3 \right)^2} &&\text{(Cauchy--Schwarz)} \\
&\leq n^{-5/2} \bar{\nu}^3 \cdot \sqrt{\sum_{i=1}^{n} \left( 2 \norm{\tilde{y}_i}_{p_1} \, \norm{\psi_i}_{q_1} \right)^6} &&\text{(Lemma~\ref{lm:nu-bound} and \ref{lm:r-bound})} \\
&\leq n^{-2} \bar{\nu}^3 \left( 2 \norm{\tilde{y}}_{\max,p_1}^n \, \norm{\psi}_{\max,q_1}^n \right)^3.
\end{align*}
Next, we examine the second term in \eqref{eq:dw-upper-bound}.
\begin{align*}
\sum_{i=1}^{n} \mathbb{E} [X_i]^4 &=  \sum_{i=1}^{n} \mathbb{E} \left[ \nu_{ni} \left(\hat{\theta}_i(z, \omega) - \theta_i(\tilde{y}_i)\right) \right]^4 \\
&= \sum_{i=1}^{n} \nu_{ni}^4 \mathbb{E} \left[ \hat{\theta}_i(z, \omega) - \theta_i(\tilde{y}_i) \right]^4 \\
&\leq \sqrt{\sum_{i=1}^{n} \nu_{ni}^8} \cdot \sqrt{\sum_{i=1}^{n} \left(\mathbb{E} \left[ \hat{\theta}_i(z, \omega) - \theta_i(\tilde{y}_i) \right]^4 \right)^2} &&\text{(Cauchy--Schwarz)} \\
&\leq n^{-7/2} \bar{\nu}^4 \cdot \sqrt{\sum_{i=1}^{n} \left( 2 \norm{\tilde{y}_i}_{p_2} \, \norm{\psi_i}_{q_2} \right)^8} &&\text{(Lemma~\ref{lm:nu-bound} and \ref{lm:r-bound})} \\
&\leq n^{-3} \bar{\nu}^4 \left( 2 \norm{\tilde{y}}_{\max,p_2}^n \, \norm{\psi}_{\max,q_2}^n \right)^4.
\end{align*}

Thus, we obtain
\begin{align}
d_W(W_n, Z) &\leq \frac{\bar{\nu}^3 D_n^2}{n^2 \sigma_n^3}  \left( 2 \norm{\tilde{y}}_{\max,p_1}^n \, \norm{\psi}_{\max,q_1}^n \right)^3 + 
\frac{\sqrt{26} \bar{\nu}^2 D_n^{3/2}}{\sqrt{\pi} n^{3/2} \sigma_n^2} \left( 2 \norm{\tilde{y}}_{\max,p_2}^n \, \norm{\psi}_{\max,q_2}^n \right)^2 \nonumber \\
&= O\left( \frac{D_n^2}{n^2 \sigma_n^3} \right) + O\left( \frac{D_n^{3/2}}{n^{3/2} \sigma_n^2} \right), \label{eq:bigO-terms}
\end{align}
under Assumption~\ref{as:bounded-norms}, which states that 
$\norm{\tilde{y}}_{\max,p_1}^n$, $\norm{\psi}_{\max,q_1}^n$, $\norm{\tilde{y}}_{\max,p_2}^n$,
and $\norm{\psi}_{\max,q_2}^n$ converge for some \( p_1, q_1, p_2, q_2 \in [1, \infty] \)
satisfying \( 1/p_1 + 1/q_1 = 1/3 \) and \( 1/p_2 + 1/q_2 = 1/4 \).

Given Assumption~\ref{as:lowerbound-sigma},
\eqref{eq:bigO-terms} can be further bounded by
\begin{equation}
    d_W(W_n, Z) \leq O\left( \frac{D_n^2}{n^{1/2}} \right) + O\left( \frac{D_n^{3/2}}{n^{1/2}} \right),
\end{equation}
which implies that, if $D_n = o(n^{1/4})$, then $d_W(W_n, Z) \to 0$.

Lemma~\ref{lm:wasserstein} shows that Wasserstein convergence implies weak convergence,
which completes the proof.
\end{proof}

\begin{lemma}\label{lm:delta-weighted}
Let $\{\zeta_i\}_{i=1}^n$ be random variables such that
\[
\mathbb{E}[\zeta_i]=0
\quad\text{and}\quad
\sup_{i\in\mathbb N}\mathbb{E}[\zeta_i^{4}]<\infty .
\]
For deterministic weights $\nu_{ni}$ satisfying Assumption~\ref{as:weights}, define
\begin{equation}\label{eq:delta-n}
\Delta_n
\;=\;
n\sum_{(i,j)\in\mathcal{E}_n}
\Bigl(
\nu_{ni}\nu_{nj} \zeta_i\zeta_j
-
\nu_{ni}\nu_{nj}\,\mathbb{E}[\zeta_i\zeta_j]
\Bigr).
\end{equation}
Then
\begin{equation}
\displaystyle
\Delta_n
\;=\;
O_p\bigl(n^{-1/2}\,D_n^{3/2}\bigr).
\end{equation}
\end{lemma}

\begin{proof}
Set
\[
X_{ij}
\;:=\;
n\,\nu_{ni}\nu_{nj}\bigl(\zeta_i\zeta_j - \mathbb{E}[\zeta_i\zeta_j]\bigr),
\qquad (i,j) \in \mathcal{E}_n,
\]
so that $\Delta_n = \sum_{(i,j) \in \mathcal{E}_n} X_{ij}$ and
\begin{equation}\label{eq:delta-sum-cov}
\mathbb{V}[\Delta_n] = \sum_{(i,j), (l,k) \in \mathcal{E}_n} \operatorname{Cov}[X_{ij}, X_{lk}].
\end{equation}

Note that $\operatorname{Cov}[X_{ij}, X_{lk}] = 0$ whenever $X_{ij}$ and $X_{lk}$ are independent, so the sum involves only a restricted number of dependent pairs.

Since $\nu_{ni} \le \bar{\nu}/n$, we have
$n\,\nu_{ni}\nu_{nj} \le \bar{\nu}^2/n$.
Then, for any dependent $(i,j), (l,k) \in \mathcal{E}_n$, we apply Cauchy--Schwarz and the uniform fourth-moment bound $\sup_i \mathbb{E}[\zeta_i^4] < \infty$ to obtain
\begin{equation}\label{eq:cov-xx}
\operatorname{Cov}[X_{ij}, X_{lk}]
\le \left( \frac{\bar{\nu}^2}{n} \right)^2
\cdot \mathbb{E}\Big|\big(\zeta_i\zeta_j - \mathbb{E}[\zeta_i\zeta_j]\big)
\big(\zeta_l\zeta_k - \mathbb{E}[\zeta_l\zeta_k]\big)\Big|
\le \frac{\bar{\nu}^4 C}{n^2}
\end{equation}
for some constant $C > 0$.
This bound also applies when $(i,j) = (l,k)$.

Fix $(i,j)\in\mathcal{E}_n$.
A second pair $(k,\ell)$ can be dependent with $(i,j)$ only if it lies within the corresponding neighbourhoods,
that is,
$(k, \ell) \cap (N_i \cup N_j) \neq \emptyset$.
Each unit has at most $D_n$ neighbours, so there are at most
\( D_n^2 - 1 =O(D_n^{2})\) such pairs.
Hence every $X_{ij}$ is dependent with at most $C_1D_n^{2}$ other summands
for some constant $C_1>0$.

There are at most $nD_n$ ordered pairs in $\mathcal{E}_n$, so
\begin{equation}\label{eq:var-delta}
\mathbb{V}[\Delta_n]
    \;\le\;
    C_1D_n^{2}\cdot nD_n\cdot \frac{\bar\nu^{4}C}{n^{2}}
    \;=\;
    C_2\,\bar\nu^{4}\,\frac{D_n^{3}}{n},
\end{equation}
for a constant $C_2>0$.

By Chebyshev's inequality, for any $\varepsilon>0$,
\[
\mathbb{P}\bigl(|\Delta_n|>\varepsilon\bigr)
    \;\le\;
    \frac{\operatorname{Var}(\Delta_n)}{\varepsilon^{2}}
    \;\le\;
    \frac{C_2\,\bar\nu^{4}}{\varepsilon^{2}}\,
    \frac{D_n^{3}}{n}.
\]
Since the right-hand side is of order
\(
n^{-1}\,D_n^{3},
\)
we have
\[
\Delta_n \;=\; O_p\bigl(n^{-1/2}\,D_n^{3/2}\bigr),
\]
which completes the proof.
\end{proof}

\begin{proof}[Proof of Theorem~\ref{th:consistency-variance}]
By Theorem~\ref{th:unbiasedness}, we have \( \mathbb{E}[\zeta_i] = 0 \), where \( \zeta_i := \hat{\theta}_i(z, \omega) - \theta_i(\tilde{y}_i) \).  
Assumption~\ref{as:ffm}\((b)\) ensures that \( \sup_{i \in \mathbb{N}} \mathbb{E}[\hat{\theta}_i(z, \omega)^4] < \infty \), which is equivalent to  
\( \sup_{i \in \mathbb{N}} \mathbb{E}[\zeta_i^4] < \infty \), since \( \theta_i(\tilde{y}_i) \) is constant.

Lemma~\ref{lm:delta-weighted} together with Assumption~\ref{as:weights} implies that
\[
n \hat{\sigma}_n^2 - n \sigma_n^2 = n \sum_{(i,j) \in \mathcal{E}_n} \nu_{ni} \nu_{nj} \left( \zeta_i \zeta_j - \mathbb{E}[\zeta_i \zeta_j] \right)= O_p(n^{-1/2} D_n^{3/2}).
\]

If Assumption~\ref{as:mns} holds with \( D_n = o(n^{1/3}) \),
\eqref{eq:var-delta} implies that this expression converges to zero in mean square.
\end{proof}

\begin{proof}[Proof of Corollary~\ref{co:normality}]
The result follows directly from Theorems~\ref{th:normality} and~\ref{th:consistency-variance}.  
We write
\begin{align*}
\hat{\sigma}_n^{-1} \left( \hat{\tau}_n(z, \omega) - \tau_n \right) = \frac{\sqrt{n}\left( \hat{\tau}_n(z, \omega) - \tau_n \right)}{\sqrt{n}\sigma_n} \cdot \frac{\sqrt{n}\sigma_n}{\sqrt{n}\hat{\sigma}_n}.
\end{align*}
The first factor converges in distribution to \( \mathcal{N}(0, 1) \) by Theorem~\ref{th:normality}, and the second factor converges to 1 in probability by Theorem~\ref{th:consistency-variance}, together with Assumption~\ref{as:lowerbound-sigma}, which guarantees that \( \sqrt{n}\sigma_n \) is bounded away from zero.

Since the product of a sequence converging in distribution and another converging in probability (to a constant) also converges in distribution, Slutsky's theorem yields the desired result.
\end{proof}

\begin{lemma}\label{lm:vvest}
$\Delta_n^c \;=\; O_p\bigl(n^{-1/2}\,D_n^{3/2}\bigr)$.
\end{lemma}

\begin{proof}[Proof of Lemma~\ref{lm:vvest}]
Define
\begin{equation}\label{eq:delta-n-c}
\Delta_n^c
\;=\;
n \sum_{(i,j) \in \mathcal{E}_n^c}
\left(
\nu_{ni} \nu_{nj} \zeta_i \zeta_j
-
\nu_{ni} \nu_{nj} \, \mathbb{E}[\zeta_i \zeta_j]
\right),
\end{equation}
so that \( n \hat{\sigma}_{cn}^2 - n \sigma_n^2 = \Delta_n^c \).
To bound the variance of \( \Delta_n^c \), observe that
\begin{equation}\label{eq:delta-sum-cov2}
\mathbb{V}[\Delta_n^c] = \sum_{(i,j), (l,k) \in \mathcal{E}_n^c} \operatorname{Cov}[X_{ij}, X_{lk}].
\end{equation}
Since \( \mathcal{E}_n^c \subseteq \mathcal{E}_n \),
\begin{equation*}
\mathbb{V}[\Delta_n] - \mathbb{V}[\Delta_n^c] = \sum_{(i,j), (l,k) \in \mathcal{E}_n /\mathcal{E}_n^c} \operatorname{Cov}[X_{ij}, X_{lk}].
\end{equation*}
This difference can be positive when $\sum_{(i,j), (l,k) \in \mathcal{E}_n /\mathcal{E}_n^c} \operatorname{Cov}[X_{ij}, X_{lk}] \geq 0$,
but this cannot be guaranteed in general, as some of the covariance terms may be negative.

Nonetheless, we can assert that the number of terms in the upper bound of $\mathbb{V}[\Delta_n^c]$ is no greater than $C_1 D_n^2$, and that the bound in \eqref{eq:var-delta} applies as well, since it is derived by summing the absolute values of the covariances.
Therefore,
\[
\Delta_n^c \;=\; O_p\bigl(n^{-1/2}\,D_n^{3/2}\bigr),
\]
as claimed.
\end{proof}

\begin{proof}[Proof of Theorem~\ref{th:consistency-variance2}]
By Lemma~\ref{lm:vvest},
$\Delta_n^c = O_p\left(n^{-1/2} D_n^{3/2} \right)$.
In particular, if Assumption~\ref{as:mns} holds with \( D_n = o(n^{1/3}) \), then \( \hat{\sigma}_{cn}^2 - \sigma_n^2 \) goes to zero in mean square.
\end{proof}

\begin{proof}[Proof of Corollary~\ref{co:normality2}]
The proof follows the same argument as in Corollary~\ref{co:normality}, replacing \( \hat{\sigma}_n^2 \) with \( \hat{\sigma}_{cn}^2 \).
\end{proof}

\begin{proof}[Proof of Theorem~\ref{th:consistency-variance3}]
Define
\begin{equation}
\tilde{\Delta}_n
\;=\;
n \sum_{(i,j) \in \tilde{\mathcal{E}}_n}
\left(
\nu_{ni} \nu_{nj} \zeta_i \zeta_j
-
\nu_{ni} \nu_{nj} \, \mathbb{E}[\zeta_i \zeta_j]
\right).
\end{equation}
Obviously, by Assumption~\ref{as:conserv}, the upper bound of $\mathbb{V}[\tilde{\Delta}_n]$ lies between those of $\mathbb{V}[\Delta_n^c]$ and $\mathbb{V}[\Delta_n]$, and the result follows directly.
\end{proof}

\begin{proof}[Proof of Corollary~\ref{co:normality3}]
The proof follows the same argument as in Corollary~\ref{co:normality}, replacing \( \hat{\sigma}_n^2 \) with \( \tilde{\sigma}_{n}^2 \).
\end{proof}

\begin{proposition}[Expected Inverse Neighbourhood Size under Blockwise Erd\H{o}s--Rényi Graphs]\label{pp:inv-degree}
Fix a unit $i$ and let $B_i$ denote the block to which unit $i$ belongs, with block size
$m_i = |B_i|$.
Assume that edges are formed independently within each block according to an Erd\H{o}s--Rényi model with edge probability $p_{i} \in (0,1)$, and that $A_{ii}(\omega)=1$ so that neighbourhoods are defined in the closed sense.
Then
\begin{equation}\label{eq:inv-degree-expectation}
\mathbb{E}_\omega\!\left[\frac{1}{|N_i(\omega)|}\right]
=
\frac{1-(1-p_{i})^{m_i}}{m_i\,p_{i}}.
\end{equation}
\end{proposition}

\begin{proof}
Since $A_{ii}(\omega)=1$, the neighbourhood size can be written as
\[
|N_i(\omega)| = 1 + X_i,
\]
where
\[
X_i = \sum_{\substack{j\in B_i\\ j\neq i}} A_{ij}(\omega)
\sim \mathrm{Binomial}(m_i-1,p_i),
\]
by independence of edges within block $B_i$.

Therefore,
\[
\mathbb{E}_\omega\!\left[\frac{1}{|N_i(\omega)|}\right]
=
\mathbb{E}\!\left[\frac{1}{1+X_i}\right].
\]
Using the identity
\[
\mathbb{E}\!\left[\frac{1}{1+X}\right]
=
\frac{1-(1-p)^{N+1}}{(N+1)p},
\qquad
X\sim\mathrm{Binomial}(N,p),
\]
with $N=m_i-1$ and $p=p_i$ yields \eqref{eq:inv-degree-expectation}.
\end{proof}

\section{Representer Estimation in Practice}\label{sec:estimation}

The algorithms below construct the Riesz representer conditional on a realised latent environment $\omega_0$.
This reflects the practical setting in which both the treatment assignment and the observed outcomes correspond to a single realised experiment.

At the theoretical level, the representer is defined as a random element over the latent space $\Omega$.
The asymptotic results in the main text explicitly allow $\psi_i$ to vary across $\omega$ and therefore do not rely on the representer being deterministic.

Conditional on a realised environment $\omega_0$, however, the representer is determined entirely by the treatment-effect functional $\theta_i$ and the randomisation design.
The resulting construction is therefore computational and carried out pointwise in $\omega_0$.

\subsection{Computation in Finite-Dimensional Model Spaces}

In finite-dimensional model spaces \(\mathcal{M}_i\), we can estimate \(\psi_i(z,\omega_0)\) via basis expansion and moment matching. Let \(\{g_{i,1}, \dots, g_{i,m}\}\) be a basis for \(\mathcal{M}_i\) and $\mathrm{dim}(\mathcal{M}_i) = m$. We solve for coefficients \(\boldsymbol{\beta}_i \in \mathbb{R}^m\) such that
\begin{equation}
    \hat{\psi}_i(z, \omega_0) = \sum_{k=1}^m \beta_{i,k} \, g_{i,k}(z, \omega_0)
\end{equation}
satisfies the identity \(\theta_i(\tilde{y}_i) = \langle \tilde{y}_i, \hat{\psi}_i \rangle^0\).
Algorithm~\ref{alg:alg1} summarises this procedure.
This construction aligns with the approach described by \citet{harshaw2022riesz} (preprint version), where Riesz representers are computed via basis expansion and moment matching under the randomisation distribution.

Throughout the construction of the Riesz representer, all objects are understood conditionally on the realised latent state \( \omega_0 \).
Accordingly, basis functions are evaluated at \( \omega_0 \) and can be treated as functions of the treatment assignment \( z \) alone.
Formally, this corresponds to working on the fibre \( \mathcal{M}_i(\omega_0) \subset L^2(\mathcal{Z}) \) induced by the Bochner space representation of RPO.

When the treatment space \( \mathcal{Z} \) is finite with \( K \) distinct levels, a simple and interpretable basis consists of treatment indicators.
In this case, each basis function is defined as a function of \( z \) that equals one if \( z \) corresponds to the \( k \)-th treatment condition and zero otherwise.
This yields a saturated approximation space with a separate coefficient for each treatment arm and does not rely on any structural assumptions about the response surface.

When \( z \) is continuous or high-dimensional, implementable approximation spaces are typically constructed using low-order basis functions of \( z \), such as polynomials or splines.
These bases offer a balance between expressiveness and computational tractability and can be tailored to the structure of the experimental design.
Although the underlying potential outcome functions may depend jointly on \( z \) and the latent environment \( \omega \), this dependence is fixed at \( \omega_0 \) during the representer construction stage and is therefore absorbed into the realised section of the function space.

This procedure is fully offline in the sense that it depends only on the design distribution for \(z\), the basis functions, and the known functional \(\theta_i\).
It does not require access to observed outcomes \(Y_i = \tilde{y}_i(z_0,\omega_0)\).

The construction relies on Assumption~\ref{as:randomisation}, which ensures that the treatment assignment distribution remains equal to the known design distribution conditional on the realised environment $\omega_0$.
This permits expectations with respect to \(z\) to be computed under the randomisation measure \(\mu\).

In symmetric settings with common \(\mathcal{M}_i\) and \(\theta_i\), the same representer \(\psi\) can be reused across units.

\subsection{Approximation in Infinite-Dimensional Model Spaces}\label{subsec:riesz-infinite}

In many practical applications, the model space \( \mathcal{M}_i \) is infinite-dimensional. 
This arises naturally when either the treatment assignment \( z \) or the latent variable \( \omega \) lies in a continuous domain, or when potential outcomes depend on complex functional relationships. 
To compute the Riesz representer in such settings, we project the problem onto a sequence of growing finite-dimensional subspaces.

The estimation procedure again proceeds conditionally on a realised environment $\omega_0$.
Accordingly, the relevant inner product is taken over the treatment assignment distribution conditional on $\omega_0$, yielding a standard projection problem in a Hilbert space.
The computation therefore reduces to constructing finite-dimensional approximations of the representer through projected moment equations.

\begin{assumption}[Separable Model Space]\label{as:separable}
The model space \( \mathcal{M}_i \subset L^2(\mathcal{Z} \times \Omega) \), equipped with the inner product \( \langle \cdot, \cdot \rangle^0 \) defined in~\eqref{eq:inner-product0} for some $\omega_0 \in \Omega$, is a separable Hilbert space.
\end{assumption}

By Assumption~\ref{as:separable}, the space \( \mathcal{M}_i \) admits a countable orthonormal basis \( \{e_k\}_{k \ge 1} \).
For each dimension \( m \in \mathbb{N} \), define the subspace
\begin{equation}
    \mathcal{M}_i^{(m)} = \operatorname{span} \{e_1, \dots, e_m\},
\end{equation}
and let \( P_m: \mathcal{M}_i \to \mathcal{M}_i^{(m)} \) denote the orthogonal projection.
Then the projected representer
\begin{equation}
    \psi_i^{(m)} := P_m \psi_i = \sum_{k=1}^{m} \langle \psi_i, e_k \rangle^0 \, e_k
\end{equation}
converges to \( \psi_i \) in norm as \( m \to \infty \), by the completeness of the Hilbert space \( \mathcal{M}_i \), since \( \{e_k\}_{k \geq 1} \) forms a total orthonormal system and \( P_m \) is the orthogonal projection onto the span of the first \( m \) basis elements.

The estimation of \( \psi_i^{(m)} \) proceeds by solving the finite-dimensional moment equation system that results from projecting the Riesz representation identity onto the span of the first \( m \) basis functions. This yields a linear system involving a Gram matrix and a target vector, as shown in Algorithm~\ref{alg:alg2}.

This procedure is still fully offline in the sense that
both \( G_i^{(m)} \) and \( T_i^{(m)} \) can be computed offline from the design distribution and the known linear functional \( \theta_i \).
It does not require access to observed outcomes \(Y_i = \tilde{y}_i(z_0, \omega_0)\).
In symmetric settings with common \(\mathcal{M}_i\) and \(\theta_i\), the same representer \(\psi\) can be reused across units.

\comm{Remark (Consistency of Representer Estimation).}

The present paper focuses on inference for expectation-based causal estimands under stochastic dependence rather than on optimal nonparametric estimation of infinite-dimensional representers.
The estimation procedure above follows a standard sieve approximation framework based on finite-dimensional projections and empirical moment equations; see, for example, \citet{chen2007sieve, chen_pouzo2012moment, van2000asymptotic}.

Under standard regularity conditions, including sufficiently regular bases and well-conditioned Gram matrices, the projected estimator converges to the population representer.
For completeness, a typical convergence rate takes the form
\begin{equation}
    \| \widehat{\psi}_i^{(m)} - \psi_i \|
    =
    O_p\left( m^{-s} + \sqrt{m/n} \right)
\end{equation}
for some \( s > 0 \), balancing approximation bias and estimation variance.

The orthonormal system $\{e_k\}$ can either be specified \emph{a priori}, using classical bases such as Fourier series, wavelets, or splines, or derived from the data through methods like principal components or kernel eigenfunctions.

\section{Algorithms}

\begin{algorithm}
\footnotesize
\caption{Computing the Riesz representer \(\psi_i(z, \omega_0)\)}\label{alg:alg1}
\begin{algorithmic}
\State \textbf{Inputs:}
\State $\bullet$ Basis functions $\{g_{i,1}, \dots, g_{i,m}\}$ for \(\mathcal{M}_i\)
\State $\bullet$ Known linear functional \(\theta_i\)
\State $\bullet$ Randomisation distribution for \(z\)

\Statex
\State \textbf{Step 1: Compute Gram Matrix \(\hat{\mathbf{G}}_i\)}
\[
[\hat{\mathbf{G}}_i]_{\ell,k} = \mathbb{E}_z\bigl[g_{i,\ell}(z, \omega_0) g_{i,k}(z, \omega_0)\bigr]
\]

\Statex
\State \textbf{Step 2: Compute Target Vector \(\hat{\mathbf{T}}_i\)}
\[
[\hat{\mathbf{T}}_i]_\ell = \theta_i(g_{i,\ell})
\]

\Statex
\State \textbf{Step 3: Solve Linear System}
\[
\hat{\mathbf{G}}_i \boldsymbol{\beta}_i = \hat{\mathbf{T}}_i \quad \Rightarrow \quad \boldsymbol{\beta}_i = \hat{\mathbf{G}}_i^{-1} \hat{\mathbf{T}}_i
\]

\Statex
\State \textbf{Output:} Estimated Riesz representer \(\hat{\psi}_i(z,\omega_0)\)
\end{algorithmic}
\end{algorithm}

\begin{algorithm}
\footnotesize
\caption{Approximating the Riesz representer \(\psi_i(z, \omega_0)\) in infinite-dimensional settings}\label{alg:alg2}
\begin{algorithmic}
\State \textbf{Inputs:}
\State $\bullet$ Truncated orthonormal basis $\{e_1, \dots, e_m\}$ of $\mathcal{M}_i$
\State $\bullet$ Known linear functional \(\theta_i\)
\State $\bullet$ Randomisation distribution for $z$

\Statex
\State \textbf{Step 1: Compute Gram Matrix \(\hat{\mathbf{G}}_i^{(m)}\)}
\[
[\hat{\mathbf{G}}_i^{(m)}]_{\ell,k} = \mathbb{E}_z\bigl[e_\ell(z, \omega_0)\, e_k(z, \omega_0)\bigr]
\]

\Statex
\State \textbf{Step 2: Compute Target Vector \(\hat{\mathbf{T}}_i^{(m)}\)}
\[
[\hat{\mathbf{T}}_i^{(m)}]_\ell = \theta_i(e_\ell)
\]

\Statex
\State \textbf{Step 3: Solve Regularised Linear System}
\[
\hat{\mathbf{G}}_i^{(m)} \boldsymbol{\beta}_i = \hat{\mathbf{T}}_i^{(m)} \quad \Rightarrow \quad \boldsymbol{\beta}_i = \left(\hat{\mathbf{G}}_i^{(m)}\right)^{-1} \hat{\mathbf{T}}_i^{(m)}
\]

\Statex
\State \textbf{Output:} Estimated Riesz representer
\[
\hat{\psi}_i^{(m)}(z, \omega_0) = \sum_{k=1}^{m} \beta_{i,k} \, e_k(z, \omega_0)
\]
\end{algorithmic}
\end{algorithm}

\end{document}

%% file: preamble.ltx
\usepackage[margin=1in]{geometry}
\usepackage{parskip}
\usepackage{setspace}
\usepackage[scaled=0.92]{helvet} 

\usepackage[utf8]{inputenc}  
\usepackage[T1]{fontenc}     
\usepackage[english]{babel}

\usepackage{fancyhdr}
\usepackage{amsmath, amssymb, amsthm}
\usepackage{mathtools}
\usepackage{mathrsfs}
\usepackage{bm}

\newtheorem{theorem}{Theorem}
\newtheorem{assumption}{Assumption}
\newtheorem{lemma}{Lemma}
\newtheorem{proposition}{Proposition}
\newtheorem{corollary}{Corollary}
\newtheorem{definition}{Definition}

\newtheorem{example}{Example}

\newcommand{\dif}{\mathop{}\!\mathrm{d}}
\newcommand{\innp}[2]{\left\langle #1,\, #2 \right\rangle}
\newcommand{\abs}[1]{\left\lvert #1 \right\rvert}
\newcommand{\norm}[1]{\left\lVert #1 \right\rVert}

\usepackage{graphicx}
\usepackage{float}
\usepackage{caption}
\usepackage{subcaption}
\usepackage{booktabs}
\usepackage{multirow}

\usepackage{tikz}
\usepackage{pgfplots}
\pgfplotsset{compat=1.14}

\usepackage[authoryear]{natbib}
\usepackage[colorlinks=true, linkcolor=blue, citecolor=blue, urlcolor=blue]{hyperref}

\usepackage{algorithm}
\usepackage{algpseudocode}

\usepackage{orcidlink}
\usepackage{xcolor}

\newcommand{\comm}[1]{}
\newcommand{\todo}[1]{}

\usepackage{enumitem}
\usepackage{comment}

%% file: plots/coverage_multiN.tex
\begin{tikzpicture}[x=1pt,y=1pt]
\definecolor{fillColor}{RGB}{255,255,255}
\path[use as bounding box,fill=fillColor,fill opacity=0.00] (0,0) rectangle (419.17,231.26);
\begin{scope}
\path[clip] (  0.00,  0.00) rectangle (419.17,231.26);
\definecolor{fillColor}{RGB}{255,255,255}

\path[fill=fillColor] ( -0.00,  0.00) rectangle (419.17,231.26);
\end{scope}
\begin{scope}
\path[clip] ( 30.59, 42.44) rectangle (409.17,184.81);
\definecolor{drawColor}{gray}{0.92}

\path[draw=drawColor,line width= 0.3pt,line join=round] ( 30.59, 70.48) --
	(409.17, 70.48);

\path[draw=drawColor,line width= 0.3pt,line join=round] ( 30.59,113.62) --
	(409.17,113.62);

\path[draw=drawColor,line width= 0.3pt,line join=round] ( 30.59,156.77) --
	(409.17,156.77);

\path[draw=drawColor,line width= 0.6pt,line join=round] ( 30.59, 48.91) --
	(409.17, 48.91);

\path[draw=drawColor,line width= 0.6pt,line join=round] ( 30.59, 92.05) --
	(409.17, 92.05);

\path[draw=drawColor,line width= 0.6pt,line join=round] ( 30.59,135.19) --
	(409.17,135.19);

\path[draw=drawColor,line width= 0.6pt,line join=round] ( 30.59,178.34) --
	(409.17,178.34);

\path[draw=drawColor,line width= 0.6pt,line join=round] ( 74.27, 42.44) --
	( 74.27,184.81);

\path[draw=drawColor,line width= 0.6pt,line join=round] (147.08, 42.44) --
	(147.08,184.81);

\path[draw=drawColor,line width= 0.6pt,line join=round] (219.88, 42.44) --
	(219.88,184.81);

\path[draw=drawColor,line width= 0.6pt,line join=round] (292.68, 42.44) --
	(292.68,184.81);

\path[draw=drawColor,line width= 0.6pt,line join=round] (365.48, 42.44) --
	(365.48,184.81);
\definecolor{drawColor}{RGB}{0,0,0}
\definecolor{fillColor}{gray}{0.85}

\path[draw=drawColor,line width= 0.6pt,fill=fillColor] ( 49.70,-1935.70) rectangle ( 60.62,125.49);

\path[draw=drawColor,line width= 0.6pt,fill=fillColor] (122.50,-1935.70) rectangle (133.42,129.80);

\path[draw=drawColor,line width= 0.6pt,fill=fillColor] (195.31,-1935.70) rectangle (206.23,135.19);

\path[draw=drawColor,line width= 0.6pt,fill=fillColor] (268.11,-1935.70) rectangle (279.03,131.96);

\path[draw=drawColor,line width= 0.6pt,fill=fillColor] (340.91,-1935.70) rectangle (351.83,114.70);
\definecolor{fillColor}{gray}{0.65}

\path[draw=drawColor,line width= 0.6pt,fill=fillColor] ( 62.44,-1935.70) rectangle ( 73.36,136.27);

\path[draw=drawColor,line width= 0.6pt,fill=fillColor] (135.24,-1935.70) rectangle (146.17,143.82);

\path[draw=drawColor,line width= 0.6pt,fill=fillColor] (208.05,-1935.70) rectangle (218.97,121.17);

\path[draw=drawColor,line width= 0.6pt,fill=fillColor] (280.85,-1935.70) rectangle (291.77,130.88);

\path[draw=drawColor,line width= 0.6pt,fill=fillColor] (353.65,-1935.70) rectangle (364.57,131.96);
\definecolor{fillColor}{gray}{0.45}

\path[draw=drawColor,line width= 0.6pt,fill=fillColor] ( 75.18,-1935.70) rectangle ( 86.10,131.96);

\path[draw=drawColor,line width= 0.6pt,fill=fillColor] (147.99,-1935.70) rectangle (158.91,133.04);

\path[draw=drawColor,line width= 0.6pt,fill=fillColor] (220.79,-1935.70) rectangle (231.71,122.25);

\path[draw=drawColor,line width= 0.6pt,fill=fillColor] (293.59,-1935.70) rectangle (304.51,124.41);

\path[draw=drawColor,line width= 0.6pt,fill=fillColor] (366.39,-1935.70) rectangle (377.31,135.19);
\definecolor{fillColor}{gray}{0.25}

\path[draw=drawColor,line width= 0.6pt,fill=fillColor] ( 87.92,-1935.70) rectangle ( 98.84,117.94);

\path[draw=drawColor,line width= 0.6pt,fill=fillColor] (160.73,-1935.70) rectangle (171.65,127.64);

\path[draw=drawColor,line width= 0.6pt,fill=fillColor] (233.53,-1935.70) rectangle (244.45,127.64);

\path[draw=drawColor,line width= 0.6pt,fill=fillColor] (306.33,-1935.70) rectangle (317.25,130.88);

\path[draw=drawColor,line width= 0.6pt,fill=fillColor] (379.13,-1935.70) rectangle (390.06,143.82);

\path[draw=drawColor,line width= 0.6pt,dash pattern=on 4pt off 4pt ,line join=round] ( 30.59,113.62) -- (409.17,113.62);
\end{scope}
\begin{scope}
\path[clip] (  0.00,  0.00) rectangle (419.17,231.26);
\definecolor{drawColor}{gray}{0.30}

\node[text=drawColor,anchor=base east,inner sep=0pt, outer sep=0pt, scale=  0.88] at ( 25.64, 45.88) {0.92};

\node[text=drawColor,anchor=base east,inner sep=0pt, outer sep=0pt, scale=  0.88] at ( 25.64, 89.02) {0.94};

\node[text=drawColor,anchor=base east,inner sep=0pt, outer sep=0pt, scale=  0.88] at ( 25.64,132.16) {0.96};

\node[text=drawColor,anchor=base east,inner sep=0pt, outer sep=0pt, scale=  0.88] at ( 25.64,175.31) {0.98};
\end{scope}
\begin{scope}
\path[clip] (  0.00,  0.00) rectangle (419.17,231.26);
\definecolor{drawColor}{gray}{0.30}

\node[text=drawColor,anchor=base,inner sep=0pt, outer sep=0pt, scale=  0.88] at ( 74.27, 31.43) {0};

\node[text=drawColor,anchor=base,inner sep=0pt, outer sep=0pt, scale=  0.88] at (147.08, 31.43) {0.1};

\node[text=drawColor,anchor=base,inner sep=0pt, outer sep=0pt, scale=  0.88] at (219.88, 31.43) {0.2};

\node[text=drawColor,anchor=base,inner sep=0pt, outer sep=0pt, scale=  0.88] at (292.68, 31.43) {0.25};

\node[text=drawColor,anchor=base,inner sep=0pt, outer sep=0pt, scale=  0.88] at (365.48, 31.43) {0.3};
\end{scope}
\begin{scope}
\path[clip] (  0.00,  0.00) rectangle (419.17,231.26);
\definecolor{drawColor}{RGB}{0,0,0}

\node[text=drawColor,anchor=base,inner sep=0pt, outer sep=0pt, scale=  1.10] at (219.88, 12.14) {Dependency growth rate $d$};
\end{scope}
\begin{scope}
\path[clip] (  0.00,  0.00) rectangle (419.17,231.26);
\definecolor{drawColor}{RGB}{0,0,0}

\node[text=drawColor,anchor=base west,inner sep=0pt, outer sep=0pt, scale=  1.00] at (132.10,205.09) {$n=$};
\end{scope}
\begin{scope}
\path[clip] (  0.00,  0.00) rectangle (419.17,231.26);
\definecolor{drawColor}{RGB}{0,0,0}
\definecolor{fillColor}{gray}{0.85}

\path[draw=drawColor,line width= 0.6pt,fill=fillColor] (154.87,202.02) rectangle (167.90,215.05);
\end{scope}
\begin{scope}
\path[clip] (  0.00,  0.00) rectangle (419.17,231.26);
\definecolor{drawColor}{RGB}{0,0,0}
\definecolor{fillColor}{gray}{0.65}

\path[draw=drawColor,line width= 0.6pt,fill=fillColor] (193.52,202.02) rectangle (206.55,215.05);
\end{scope}
\begin{scope}
\path[clip] (  0.00,  0.00) rectangle (419.17,231.26);
\definecolor{drawColor}{RGB}{0,0,0}
\definecolor{fillColor}{gray}{0.45}

\path[draw=drawColor,line width= 0.6pt,fill=fillColor] (232.17,202.02) rectangle (245.20,215.05);
\end{scope}
\begin{scope}
\path[clip] (  0.00,  0.00) rectangle (419.17,231.26);
\definecolor{drawColor}{RGB}{0,0,0}
\definecolor{fillColor}{gray}{0.25}

\path[draw=drawColor,line width= 0.6pt,fill=fillColor] (270.82,202.02) rectangle (283.85,215.05);
\end{scope}
\begin{scope}
\path[clip] (  0.00,  0.00) rectangle (419.17,231.26);
\definecolor{drawColor}{RGB}{0,0,0}

\node[text=drawColor,anchor=base west,inner sep=0pt, outer sep=0pt, scale=  0.88] at (174.11,205.51) {100};
\end{scope}
\begin{scope}
\path[clip] (  0.00,  0.00) rectangle (419.17,231.26);
\definecolor{drawColor}{RGB}{0,0,0}

\node[text=drawColor,anchor=base west,inner sep=0pt, outer sep=0pt, scale=  0.88] at (212.76,205.51) {200};
\end{scope}
\begin{scope}
\path[clip] (  0.00,  0.00) rectangle (419.17,231.26);
\definecolor{drawColor}{RGB}{0,0,0}

\node[text=drawColor,anchor=base west,inner sep=0pt, outer sep=0pt, scale=  0.88] at (251.41,205.51) {500};
\end{scope}
\begin{scope}
\path[clip] (  0.00,  0.00) rectangle (419.17,231.26);
\definecolor{drawColor}{RGB}{0,0,0}

\node[text=drawColor,anchor=base west,inner sep=0pt, outer sep=0pt, scale=  0.88] at (290.06,205.51) {1000};
\end{scope}
\end{tikzpicture}

%% file: plots/coverage_multiN3size.tex
\begin{tikzpicture}[x=1pt,y=1pt]
\definecolor{fillColor}{RGB}{255,255,255}
\path[use as bounding box,fill=fillColor,fill opacity=0.00] (0,0) rectangle (419.17,231.26);
\begin{scope}
\path[clip] (  0.00,  0.00) rectangle (419.17,231.26);
\definecolor{fillColor}{RGB}{255,255,255}

\path[fill=fillColor] ( -0.00,  0.00) rectangle (419.17,231.26);
\end{scope}
\begin{scope}
\path[clip] ( 30.59, 42.44) rectangle (409.17,184.81);
\definecolor{drawColor}{gray}{0.92}

\path[draw=drawColor,line width= 0.3pt,line join=round] ( 30.59, 70.48) --
	(409.17, 70.48);

\path[draw=drawColor,line width= 0.3pt,line join=round] ( 30.59,113.62) --
	(409.17,113.62);

\path[draw=drawColor,line width= 0.3pt,line join=round] ( 30.59,156.77) --
	(409.17,156.77);

\path[draw=drawColor,line width= 0.6pt,line join=round] ( 30.59, 48.91) --
	(409.17, 48.91);

\path[draw=drawColor,line width= 0.6pt,line join=round] ( 30.59, 92.05) --
	(409.17, 92.05);

\path[draw=drawColor,line width= 0.6pt,line join=round] ( 30.59,135.19) --
	(409.17,135.19);

\path[draw=drawColor,line width= 0.6pt,line join=round] ( 30.59,178.34) --
	(409.17,178.34);

\path[draw=drawColor,line width= 0.6pt,line join=round] ( 74.27, 42.44) --
	( 74.27,184.81);

\path[draw=drawColor,line width= 0.6pt,line join=round] (147.08, 42.44) --
	(147.08,184.81);

\path[draw=drawColor,line width= 0.6pt,line join=round] (219.88, 42.44) --
	(219.88,184.81);

\path[draw=drawColor,line width= 0.6pt,line join=round] (292.68, 42.44) --
	(292.68,184.81);

\path[draw=drawColor,line width= 0.6pt,line join=round] (365.48, 42.44) --
	(365.48,184.81);
\definecolor{drawColor}{RGB}{0,0,0}
\definecolor{fillColor}{gray}{0.85}

\path[draw=drawColor,line width= 0.6pt,fill=fillColor] ( 49.70,-1935.70) rectangle ( 60.62,140.59);

\path[draw=drawColor,line width= 0.6pt,fill=fillColor] (122.50,-1935.70) rectangle (133.42,143.82);

\path[draw=drawColor,line width= 0.6pt,fill=fillColor] (195.31,-1935.70) rectangle (206.23,149.22);

\path[draw=drawColor,line width= 0.6pt,fill=fillColor] (268.11,-1935.70) rectangle (279.03,148.14);

\path[draw=drawColor,line width= 0.6pt,fill=fillColor] (340.91,-1935.70) rectangle (351.83,155.69);
\definecolor{fillColor}{gray}{0.65}

\path[draw=drawColor,line width= 0.6pt,fill=fillColor] ( 62.44,-1935.70) rectangle ( 73.36,142.74);

\path[draw=drawColor,line width= 0.6pt,fill=fillColor] (135.24,-1935.70) rectangle (146.17,155.69);

\path[draw=drawColor,line width= 0.6pt,fill=fillColor] (208.05,-1935.70) rectangle (218.97,143.82);

\path[draw=drawColor,line width= 0.6pt,fill=fillColor] (280.85,-1935.70) rectangle (291.77,163.24);

\path[draw=drawColor,line width= 0.6pt,fill=fillColor] (353.65,-1935.70) rectangle (364.57,154.61);
\definecolor{fillColor}{gray}{0.45}

\path[draw=drawColor,line width= 0.6pt,fill=fillColor] ( 75.18,-1935.70) rectangle ( 86.10,147.06);

\path[draw=drawColor,line width= 0.6pt,fill=fillColor] (147.99,-1935.70) rectangle (158.91,151.37);

\path[draw=drawColor,line width= 0.6pt,fill=fillColor] (220.79,-1935.70) rectangle (231.71,145.98);

\path[draw=drawColor,line width= 0.6pt,fill=fillColor] (293.59,-1935.70) rectangle (304.51,145.98);

\path[draw=drawColor,line width= 0.6pt,fill=fillColor] (366.39,-1935.70) rectangle (377.31,150.30);
\definecolor{fillColor}{gray}{0.25}

\path[draw=drawColor,line width= 0.6pt,fill=fillColor] ( 87.92,-1935.70) rectangle ( 98.84,134.12);

\path[draw=drawColor,line width= 0.6pt,fill=fillColor] (160.73,-1935.70) rectangle (171.65,160.00);

\path[draw=drawColor,line width= 0.6pt,fill=fillColor] (233.53,-1935.70) rectangle (244.45,153.53);

\path[draw=drawColor,line width= 0.6pt,fill=fillColor] (306.33,-1935.70) rectangle (317.25,148.14);

\path[draw=drawColor,line width= 0.6pt,fill=fillColor] (379.13,-1935.70) rectangle (390.06,155.69);

\path[draw=drawColor,line width= 0.6pt,dash pattern=on 4pt off 4pt ,line join=round] ( 30.59,113.62) -- (409.17,113.62);
\end{scope}
\begin{scope}
\path[clip] (  0.00,  0.00) rectangle (419.17,231.26);
\definecolor{drawColor}{gray}{0.30}

\node[text=drawColor,anchor=base east,inner sep=0pt, outer sep=0pt, scale=  0.88] at ( 25.64, 45.88) {0.92};

\node[text=drawColor,anchor=base east,inner sep=0pt, outer sep=0pt, scale=  0.88] at ( 25.64, 89.02) {0.94};

\node[text=drawColor,anchor=base east,inner sep=0pt, outer sep=0pt, scale=  0.88] at ( 25.64,132.16) {0.96};

\node[text=drawColor,anchor=base east,inner sep=0pt, outer sep=0pt, scale=  0.88] at ( 25.64,175.31) {0.98};
\end{scope}
\begin{scope}
\path[clip] (  0.00,  0.00) rectangle (419.17,231.26);
\definecolor{drawColor}{gray}{0.30}

\node[text=drawColor,anchor=base,inner sep=0pt, outer sep=0pt, scale=  0.88] at ( 74.27, 31.43) {0};

\node[text=drawColor,anchor=base,inner sep=0pt, outer sep=0pt, scale=  0.88] at (147.08, 31.43) {0.1};

\node[text=drawColor,anchor=base,inner sep=0pt, outer sep=0pt, scale=  0.88] at (219.88, 31.43) {0.2};

\node[text=drawColor,anchor=base,inner sep=0pt, outer sep=0pt, scale=  0.88] at (292.68, 31.43) {0.25};

\node[text=drawColor,anchor=base,inner sep=0pt, outer sep=0pt, scale=  0.88] at (365.48, 31.43) {0.3};
\end{scope}
\begin{scope}
\path[clip] (  0.00,  0.00) rectangle (419.17,231.26);
\definecolor{drawColor}{RGB}{0,0,0}

\node[text=drawColor,anchor=base,inner sep=0pt, outer sep=0pt, scale=  1.10] at (219.88, 12.14) {Dependency growth rate $d$};
\end{scope}
\begin{scope}
\path[clip] (  0.00,  0.00) rectangle (419.17,231.26);
\definecolor{drawColor}{RGB}{0,0,0}

\node[text=drawColor,anchor=base west,inner sep=0pt, outer sep=0pt, scale=  1.00] at (132.10,205.09) {$n=$};
\end{scope}
\begin{scope}
\path[clip] (  0.00,  0.00) rectangle (419.17,231.26);
\definecolor{drawColor}{RGB}{0,0,0}
\definecolor{fillColor}{gray}{0.85}

\path[draw=drawColor,line width= 0.6pt,fill=fillColor] (154.87,202.02) rectangle (167.90,215.05);
\end{scope}
\begin{scope}
\path[clip] (  0.00,  0.00) rectangle (419.17,231.26);
\definecolor{drawColor}{RGB}{0,0,0}
\definecolor{fillColor}{gray}{0.65}

\path[draw=drawColor,line width= 0.6pt,fill=fillColor] (193.52,202.02) rectangle (206.55,215.05);
\end{scope}
\begin{scope}
\path[clip] (  0.00,  0.00) rectangle (419.17,231.26);
\definecolor{drawColor}{RGB}{0,0,0}
\definecolor{fillColor}{gray}{0.45}

\path[draw=drawColor,line width= 0.6pt,fill=fillColor] (232.17,202.02) rectangle (245.20,215.05);
\end{scope}
\begin{scope}
\path[clip] (  0.00,  0.00) rectangle (419.17,231.26);
\definecolor{drawColor}{RGB}{0,0,0}
\definecolor{fillColor}{gray}{0.25}

\path[draw=drawColor,line width= 0.6pt,fill=fillColor] (270.82,202.02) rectangle (283.85,215.05);
\end{scope}
\begin{scope}
\path[clip] (  0.00,  0.00) rectangle (419.17,231.26);
\definecolor{drawColor}{RGB}{0,0,0}

\node[text=drawColor,anchor=base west,inner sep=0pt, outer sep=0pt, scale=  0.88] at (174.11,205.51) {100};
\end{scope}
\begin{scope}
\path[clip] (  0.00,  0.00) rectangle (419.17,231.26);
\definecolor{drawColor}{RGB}{0,0,0}

\node[text=drawColor,anchor=base west,inner sep=0pt, outer sep=0pt, scale=  0.88] at (212.76,205.51) {200};
\end{scope}
\begin{scope}
\path[clip] (  0.00,  0.00) rectangle (419.17,231.26);
\definecolor{drawColor}{RGB}{0,0,0}

\node[text=drawColor,anchor=base west,inner sep=0pt, outer sep=0pt, scale=  0.88] at (251.41,205.51) {500};
\end{scope}
\begin{scope}
\path[clip] (  0.00,  0.00) rectangle (419.17,231.26);
\definecolor{drawColor}{RGB}{0,0,0}

\node[text=drawColor,anchor=base west,inner sep=0pt, outer sep=0pt, scale=  0.88] at (290.06,205.51) {1000};
\end{scope}
\end{tikzpicture}